%% file: main.tex
\documentclass{article}

\usepackage[pagewise]{lineno}

\usepackage{pgfplots}
\pgfplotsset{compat=1.18}
\usepackage{fullpage}
\usepackage{hyperref}
\usepackage{graphicx} 
\usepackage[utf8]{inputenc}
\usepackage[english]{babel}
\usepackage{mathtools}
\usepackage{amsmath}
\usepackage{amsfonts}
\usepackage{amssymb}
\usepackage{amsthm}
\usepackage{graphicx}
\usepackage{xcolor}
\usepackage{float}
\usepackage{tabularx}
\usepackage[normalem]{ulem}
\usepackage{tikz}
\usepackage{comment}
\usepackage{todonotes}
\usepackage{booktabs}
\usepackage{multirow}
\usepackage[export]{adjustbox}
\usepackage{float}
\usepackage{caption}
\usepackage{subcaption}

\author{Tom Baumbach, Matias Bender}
\newtheorem{definition}{Definition}[section]
\newtheorem{lemma}[definition]{Lemma}
\newtheorem{example}[definition]{Example}
\newtheorem{theorem}[definition]{Theorem}
\newtheorem{remark}[definition]{Remark}
\newtheorem{corollary}[definition]{Corollary}
\newtheorem{proposition}[definition]{Proposition}

\theoremstyle{definition}
\theoremstyle{remark}

\usepackage{algorithm}
\usepackage{algpseudocode}

\usepackage{mathtools}
\usepackage[arrow,matrix,curve]{xy}

\newcommand{\A}{\mathcal{A}}
\newcommand{\NN}{\mathbb{N}}
\newcommand{\Lin}{\mathcal{L}}

\newcommand{\CC}{\mathbb{C}}
\newcommand{\K}{\mathcal{K}}

\newcommand{\B}{\mathcal{B}}
\newcommand{\Nt}{\mathcal{N}_t}
\newcommand{\Ht}{\mathcal{H}_t}

\newcommand{\tA}{\mathcal{A}}
\newcommand{\R}{\mathbb{R}}
\newcommand{\Z}{\mathbb{Z}}

\newcommand{\RA}{\mathbb{R}[\mathcal{A}]}
\newcommand{\lf}{\Lambda}

\newcommand{\Rx}{\mathbb{R}[x_1,\dots,x_n]}

\newcommand{\Lx}{\mathbb{R}[x_1^{\pm 1},\dots,x_n^{\pm 1}]}

\newcommand{\rr}[1]{{\sqrt[\R]{#1}}}

\newcommand{\rIRA}{\sqrt[\mathbb{R}]{I}}

\newcommand{\VL}{\mathcal{V}_{\RA}}

\newcommand{\lin}{\mathrm{lin}}
\renewcommand{\ker}{\mathrm{ker}}
\newcommand{\rk}{\mathrm{rank}}
\newcommand{\vect}{\mathrm{vec}}
\newcommand{\Hom}{\mathrm{Hom}}
\newcommand{\Spec}{\mathrm{Specm}}

\title{
  The Moment Method for Computing Real Points of Sparse Polynomial Systems
}
\author{Tom Baumbach\footnote{
Technische Universität Berlin, Institut für Mathematik, Sekr. MA4-1, Straße
des 17. Juni 136, 10623 Berlin, Germany \\
Email address: \texttt{baumbach@math.tu-berlin.de}} \and Mat\'ias
Bender\footnote{Inria \& CMAP, CNRS, École Polytechnique, Institut
  polytechnique de Paris, Palaiseau, France \\
Email address: \texttt{matias.bender@inria.fr}
}}
\date{}

\begin{document}

\maketitle


\begin{abstract}
One of the most important problems in computational algebra is the computation of real solutions of a polynomial system.
In \cite{lasserre2008semidefinite}, Lasserre, Laurent, and Rostalski introduced a numerical method to compute the real radical of an ideal defining a finite set of real points.
Their method is based on moment matrices and reduces the problem to solving a semidefinite program.

In this work, we build on this method and on the work of Laurent and
Mourrain \cite{laurent2009generalized} to exploit the sparsity of
the input polynomials.
For this purpose, we extend the theory of
moment matrices to the setting of affine toric varieties.
We prove that, for sparse polynomial systems, this
approach allows us to recover real points on affine toric varieties. Moreover, we show that a projection of the spectrahedron associated with the semidefinite program is, in fact, a polytope, and we propose a variation of the method to compute only one real solution of the system.

\end{abstract}

\input{Sections/Introduction}

\input{Sections/Preliminaries}

\input{Sections/Sparse_Extension}

\input{Sections/ComputingOneRoot}

\input{Sections/Experiments}

\paragraph{Acknowledgments}
The authors thank Philipp di Dio for useful explanations on the moment
problem and Lorenzo Baldi for his aid with \texttt{MomentPolynomialOpt.jl} \cite{MomentPolynomialOpt}.
The authors acknowledge the use of AI (ChatGPT 5.6 Sol) for assistance
with the drafting of the manuscript.
The final content, analysis and conclusions
remain the sole responsibility of the authors.

This research benefited from the support of the 
FMJH Program Gaspard Monge for optimization and operations 
research and their interactions with data science.
MB was partially supported by the PGMO grant SOAP and 
ANR JCJC PeACE (ANR-25-CE48-3760).

\bibliographystyle{plain}
\bibliography{Sections/biblio}

\newpage

\input{Sections/Appendix}


\end{document}

%% file: Sections/Introduction.tex
\section{Introduction}

Nonlinear algebra is the study of curved geometric objects through the
lens of algebra.
Drawing an analogy with linear algebra, polynomials replace linear
forms and varieties replace vector spaces.
These objects are much more expressive than linear ones, allowing us
to tackle more general problems.
The central computational problem in nonlinear algebra is how to solve
polynomial systems, that is, how to compute the common zeros of a
system of polynomial equations.
However, the meaning of solving depends on the context.
On the one hand, we might be interested in special kinds of solutions,
such as complex, real, positive, or rational solutions.
On the other hand, we might be interested in numerical approximations
of the solutions or in exact representations of them.
Which meaning is appropriate depends on the application, and there is
a broad amount of bibliography addressing these different problems.
We refer to \cite{dickenstein_solving_2005} and the references therein
for a broad introduction to different solving methods.

\paragraph{Real solving}
In this paper we are concerned with finding real solutions, and
especially positive ones.
This problem arises in several contexts such as chemical reaction networks, dynamical systems, kinematics, statistics, and game
theory; see for example \cite[Chapter~9]{sommese2005numerical} and references therein.
Arguably, this is among the hardest computational problems that one can hope to solve \emph{efficiently}.
Generally, to solve this problem, we consider its complex version and
then extract the real solutions from the complex ones.
This has several drawbacks. In particular, the number of complex
solutions can be much larger than the number of real, or positive, ones.
For a generic square dense system, that is, a system with the same
number of equations and unknowns, the number of isolated complex solutions,
counted with multiplicity, is given by B\'ezout's theorem as the product of
the degrees of the equations.
In contrast, the number of real solutions of random systems depends
strongly on the distribution of the coefficients. For instance, the
expected number of real roots grows as the square root of the number
of complex roots in some models and only logarithmically in others;
see \cite{edelman1995many} and the references therein.
Moreover, in several applications, such as some classes of chemical
reaction networks, the objective is only to compute one positive solution, and
structural arguments may guarantee that it is unique
\cite{feinberg1995existence}.
In such cases, computing all the real solutions performs substantially
more work than required.

From an algorithmic point of view, interval Newton methods, going
back to Krawczyk \cite{krawczyk1969newton}, and Smale's $\alpha$-theory
\cite{smale1986newton} provide techniques to certify approximate zeros.
In particular, modern implementations can certify that the exact zero
associated with a complex approximation is real
\cite{hauenstein2012alphacertified,breiding2023certifying}.
Related symbolic--numeric methods construct certified Hermite matrices from
approximate roots and use their signatures to certify the existence of a real
root near a given point \cite{akoglu2023certified}.
For systems depending on parameters, Hermite matrices can also describe
semialgebraic regions of the parameter space on which the number of real
solutions is constant \cite{le2022parametric}.
A recent alternative approach to this classification problem uses generic
rational univariate representations and Sturm--Habicht sequences
\cite{corniquel2026parametric}.

Let us also mention that homotopy techniques based on Viro's
patchworking method allow us to track only real solutions in certain
situations \cite{ergur2023polyhedral}, but their applicability is constrained to special families of systems.
In general, even when the real locus of a variety is finite, its
complex locus may have positive dimension, making methods that first compute
all complex solutions unsuitable.
Special techniques were developed to work around this issue, for
example critical point methods
\cite{bank2001polar,basu2006algorithms}, at the cost of auxiliary
computations over the complex numbers.
These methods are designed to recover the real locus and do not
directly reduce the computation when only one positive point is required.

Among the general methods for computing real solutions, three
families are particularly relevant here.
First, we can use quantifier elimination methods such as Cylindrical
Algebraic Decomposition (CAD) to find the solutions.
CAD is a very powerful and general method that can be used to solve
most real algebraic problems \cite{collins1976quantifier}.
However, because of its generality, its complexity is doubly
exponential in the number of variables, e.g., \cite[Chapter~11.4]{basu2006algorithms}, and in practice this method
can be too slow and consume too much memory for our applications.

Second, we have subdivision methods, e.g.,~\cite{sherbrooke1993computation,mourrain2009subdivision,mantzaflaris2011continued}, which rely on decomposing the
domain of interest into smaller regions and analyzing them
independently.
Starting from an initial box containing the solutions, the method
repeatedly subdivides the domain and applies local tests to determine
whether a given region can contain a solution.
These tests typically rely on interval arithmetic, Bernstein
representations, or variants of Descartes' rule, allowing one to
certify that a box contains no root or that it contains a unique one \cite{moore_interval_2009}.
If the test is inconclusive, the region is subdivided further.
Subdivision methods are global procedures: they can guarantee that all
solutions in a prescribed domain are found, in contrast with purely
local methods such as Newton's method.
Though powerful, these methods suffer when computing singular
solutions and are affected by the curse of dimensionality, since the
number of regions grows exponentially with the dimension.
They may benefit from the faster evaluation of sparse polynomials, but they do not otherwise exploit the algebraic structure induced by the supports.

Third, we have methods based on moments.
Introduced in \cite{lasserre2008semidefinite}, 
these methods allow us to approximate the real radical of an ideal,
that is, the ideal of every polynomial vanishing on the real locus,
when this real locus is finite.
From this real radical it is possible to approximate the real
solutions, even when the complex ones are infinite.
The idea is to construct a linear functional on
$\R[x_1,\dots,x_n]$ whose moment matrix has the real radical of the
input ideal as its kernel.
For this, one considers linear functionals annihilating the ideal and
selects the \emph{positive} ones, that is, those which are nonnegative on
squares.
This condition translates into the positivity of a moment matrix,
which must be positive semidefinite.
To obtain such an functional, the method constructs finite
  truncations of it.
  Each truncation is obtained from a point of maximal rank in a
  spectrahedron defined by these constraints, which can be
  approximated by semidefinite programming \cite{SDP}.
The kernel of the corresponding moment matrix then provides
approximations of equations defining the real radical
\cite{lasserre2008semidefinite,baldi2021computing}.
When the real solutions are finite, using the flat extension theorem
\cite{curto1998flat}, one can approximate the real
solutions by reducing the problem to an eigenvalue
problem \cite{lasserre2013moment}.
The available degree bounds for real radical computation are,
however, extremely pessimistic \cite{becker1993computation}.
In practice, the method can sometimes be applied successfully, but it
may lead to numerical inaccuracies and very large semidefinite
programs.
The objective of this work is to improve the situation in both
directions.

\paragraph{Structured systems}
As nonlinear algebra allows us to tackle harder problems, it also has
an important drawback: polynomials are much harder to compute with
than linear systems, which restricts their practical applicability.
For this reason, it is of central interest to develop algorithms that
exploit the inner structure of the problems of interest, in order to
make these methods applicable to larger systems.
One of the most common structures is sparsity: only a small subset of
the possible monomials appears in the input polynomials.
Over the complex numbers, this structure is exploited by sparse
elimination and resultant methods \cite{emiris1996sparse,bender2022toric}, sparse Gr\"obner
basis algorithms \cite{faugere2014sparse,bender2019grobner}, and polyhedral homotopy
continuation \cite{huber1995polyhedral}, as well as by geometric resolution
methods for sparse systems \cite{herrero2010computing}. See also
\cite{bender2022solving} and the references therein.
In general, these techniques are based on toric
geometry, which provides a geometric framework that relates sparse polynomial
systems to the combinatorics of their Newton polytopes
\cite{cox2024toric}.

However, as explained above, there may still be a large gap between
the numbers of complex and real, or positive, solutions.
For a generic square system, the number of isolated solutions in the
complex torus, counted with multiplicity, is the mixed volume of the Newton
polytopes of the equations by Bernstein's theorem
\cite{bernshtein_number_1975}.
For positive solutions, Khovanskii's fewnomial theory gives bounds in
terms of the number of monomials \cite{khovanskiui1991fewnomials}, and sharper
bounds for the expected number of real solutions are known in some random
models, e.g.,~\cite{burgisser2019number}.

To the best of our knowledge, the moment approach has not been
generalized to the setting of toric varieties in order to exploit
sparsity. The closest related result is the generalized flat extension
theorem of Laurent and Mourrain \cite{laurent2009generalized}. It
allows moment matrices to be indexed by sets of monomials connected to
one and therefore already reduces the set of monomials under
consideration. However, it is formulated in a polynomial ring and does
not take into account the relations among monomials in a semigroup
algebra. In this work, we incorporate these relations through the
toric coordinate ring and relate the moment construction directly to
the combinatorics of the supports of the input polynomials.

\paragraph*{Contributions}
In this paper, we extend the moment approach to the semigroup algebra
$\RA$ associated with a finite monomial support $\A$. Its real maximal spectrum
is the real part of an affine toric variety. We explain how its points are
related to solutions in the real torus and give a lattice condition under which
the two notions agree; see Section~\ref{sec:realtoricvar}, in particular
Proposition~\ref{prop:solsOnTorus}.

First, we characterize the finite-rank moment matrices over $\RA$ and
show that their kernels are real radical precisely when the corresponding
functional is a finite linear combination of evaluations
(Theorem~\ref{TheoremFiniteRankMomentMatrix}). We also show that the normalized
positive functionals annihilating the input ideal form the convex hull of the
evaluations at its real toric points (Theorem~\ref{thm:K_is_polytope}). We then
prove a sparse flat extension theorem
(Theorem~\ref{TheoremSparseFlatExtension}). In contrast with the generalized
flat extension theorem of Laurent and Mourrain, the truncation must also
contain generators of the toric ideal; this requirement is measured by the
constant $\rho_\A$; see \eqref{eq:rhoA}.

Second, we use these results to extend the real radical algorithm of
Lasserre, Laurent, and Rostalski to the toric setting. We prove that the kernels
obtained from maximal-rank points of the truncated spectrahedra eventually
generate the real radical in $\RA$
(Theorem~\ref{TheoremNtEqRealRadicalI}), give rank conditions which certify
this event (Theorem~\ref{TheoremNtEqualsRealRadicalWithRankCondition}), and
obtain a terminating algorithm for computing all the real toric points
(Algorithm~\ref{AlgorithmMomentApproachForLx} and
Theorem~\ref{TheoremAlgoLxTerminates}).

Third, we study how to compute only one real toric point. We prove
that, after projecting onto moment coordinates of fixed degree, the truncated
spectrahedra eventually coincide with the simplex whose vertices are the
truncated evaluations at the real toric points, even though the spectrahedra
themselves need not be polyhedra (Theorem~\ref{thm:KtIsPolytope}). This leads
to a randomized algorithm which terminates almost surely and returns every real
toric point with positive probability
(Algorithm~\ref{AlgorithmOneRoot} and
Theorem~\ref{thm:one-root-correctness}). We relate this probability to the
solid angle of the corresponding normal cone and discuss the case in which
there is a unique solution with nonnegative toric coordinates
(Proposition~\ref{prop:unique-nonnegative-solution}).

Finally, we compare the dense and sparse moment constructions, and
the computation of all solutions with that of one solution, on several
examples. The experiments indicate that the sparse construction produces
smaller semidefinite programs and can improve running time and numerical
robustness. For the one-solution variant, the observed improvement in accuracy
is more modest and depends on the semidefinite programming solver; see
Section~\ref{SectionComputationalResults}.

To make the our results self-contained, we give complete proofs
whenever the arguments must be adapted to the semigroup algebra, even
when they closely follow their classical polynomial-ring
counterparts. When an argument carries over without change, we instead
refer to the original treatments, principally those of Curto and
Fialkow \cite{curto1998flat}, Lasserre, Laurent, and Rostalski
\cite{lasserre2008semidefinite}, Laurent and Rostalski
\cite{laurent2012approach}, and Laurent and Mourrain
\cite{laurent2009generalized}.


%% file: Sections/Preliminaries.tex
\section{Preliminaries}
\label{Preliminaries}

This section establishes the algebraic framework and notation used throughout this paper. We assume familiarity with the fundamental theory of commutative rings and fields. Standard references for this material include \cite{AlgebraLang} and \cite{UsingAlgGeo}.

Let $R$ be a commutative ring with unity. For a subset $S \subseteq R$, we denote by $\langle S \rangle$ the ideal generated by $S$.
An ideal $\mathfrak{m} \subsetneq R$ is \emph{maximal} if there exists no ideal $J$ such that $\mathfrak{m} \subsetneq J \subsetneq R$. Equivalently, $\mathfrak{m}$ is maximal if and only if the quotient ring $R/\mathfrak{m}$ is a field. By Zorn's Lemma, every proper ideal is contained in a maximal ideal.
An ideal $I \subseteq R$ is called \emph{zero-dimensional} if the Krull dimension of the quotient ring $R/I$ is zero.
A central object of interest in this paper is the class of real radical ideals.
The \emph{real radical} of an ideal $I \subseteq R$, denoted by $\rIRA$, is defined as:
\begin{equation*}
    \rIRA \coloneqq \left\{ f \in R \ \middle|\ f^{2m} + \sum_{j} h_j^2 \in I \text{ for some } m \ge 1, h_j \in R \right\}.
\end{equation*}
An ideal $I$ is called \emph{real radical} if $\rIRA = I$.
Real radical ideals satisfy the following crucial property, often used as an alternative definition (e.g., in fields of characteristic zero):

\begin{lemma}[{\cite[Proposition 4.1.7]{bochnak2013real}}]
    An ideal $I$ is real radical if and only if for all $p_i \in R$:
    \[
        \sum_i p_i^2 \in I \implies p_i \in I \quad \text{for all } i.
    \]
\end{lemma}

We now relate ideals to their geometric solution sets.
Let $K$ be a field and $x_1, \dots, x_n$ be indeterminates.
A \emph{monomial} is a product of the form $x^\alpha \coloneqq x_1^{\alpha_1} \cdots x_n^{\alpha_n}$, where the exponent vector $\alpha = (\alpha_1, \dots, \alpha_n)$ lies in $\mathbb{Z}^n$. The \emph{total degree} of a monomial is defined as $\deg(x^\alpha) \coloneqq |\alpha| = \sum_{i=1}^n |\alpha_i|$.
We denote the ring of polynomials by $K[\textbf{x}] \coloneqq K[x_1, \dots, x_n]$.
The ring of Laurent polynomials $K[x_1^{\pm 1}, \dots, x_n^{\pm 1}]$, consists of $K$-linear combinations of monomials with integer exponents.
Note that $K[x_1^{\pm 1}, \dots, x_n^{\pm 1}]$ can be viewed as the localization of the polynomial ring $K[x]$ at the monomial $x_1 \cdots x_n$. In the context of polynomial rings $R$ over a field $K$, this is equivalent to $R/I$ being a finite-dimensional vector space over $K$.
For a set of polynomials $F = \{f_1, \dots, f_\mu\} \subset K[x_1,\dots,x_n]$, the \emph{affine variety} defined by $F$ is the set of common zeros:
\[
    \mathcal{V}(F) \coloneqq \left\{ \mathbf{a} \in K^n \mid f(\mathbf{a}) = 0 \text{ for all } f \in F \right\}.
\]
If $I = \langle f_1, \dots, f_\mu \rangle$, then $\mathcal{V}(f_1, \dots, f_\mu) = \mathcal{V}(I)$.
Conversely, for any subset $V \subseteq K^n$, the \emph{vanishing ideal} of $V$ is defined as follows:
    \[
        \mathcal{I}(V) \coloneqq \left\{ f \in K[x] \mid f(\mathbf{a}) = 0 \text{ for all } \mathbf{a} \in V \right\}.
    \]
When working over $\mathbb{R}$, we are interested in real solutions. Let $I \subseteq \mathbb{R}[x]$ be an ideal. The \emph{real variety} associated to $I$ is:
\[
    \mathcal{V}_{\mathbb{R}}(I) \coloneqq \{ \mathbf{a} \in \mathbb{R}^n \mid f(\mathbf{a}) = 0 \ \forall f \in I \}.
\]
When working with Laurent polynomials, we restrict evaluation to nonzero coordinates. For this purpose, we define the algebraic torus $(K^*)^n$ where $K^* \coloneqq K \setminus \{0\}$.
The connection between the geometry of the real variety and the algebra of the ideal is given by the \emph{Real Nullstellensatz}. This theorem generalizes Hilbert's classical Nullstellensatz to the real setting.

\begin{theorem}[Real Nullstellensatz {\cite[Theorem 4.1.4]{bochnak2013real}}]
    Let $I$ be an ideal in $\mathbb{R}[\textbf{x}]$. Then the vanishing ideal of the real variety is the real radical of $I$:
    \[
        \mathcal{I}(\mathcal{V}_{\mathbb{R}}(I)) = \rIRA.
    \]
\end{theorem}

Finally, we state a result concerning the finiteness of the solution set, which is central to our computational approach.

\begin{theorem}[{\cite[Proposition 7]{becker1993computation}}]
    Let $I$ be an ideal in $\mathbb{R}[\textbf{x}]$. The real variety $V_{\mathbb{R}}(I)$ is a finite set if and only if the real radical ideal $\rIRA$ is zero-dimensional.
    Furthermore, the number of real solutions is given by:
    \[
        |\mathcal{V}_{\mathbb{R}}(I)| = \dim_{\mathbb{R}} \left( \mathbb{R}[x] / \rIRA \right).
    \]
\end{theorem}

\subsection{Toric varieties and their real part} \label{sec:realtoricvar}

In this subsection, we briefly introduce toric varieties and their real part.

\begin{definition}
\label{Def:Coordinate_Ring_of_A}
Let $\A \subset \Z^n$ be a finite set. For $s \in \mathbb{N}$ we define the sets 
\begin{equation*}
\tA_s \coloneqq \Bigl\{ \: \sum_{i= 1}^m  \alpha_i  \: | \: m \leq s, \alpha_i \in \A  \Bigl\} \quad \text{ and } \quad \A_\infty \coloneqq \lin_\NN(\mathcal{A}) = \left\{ \sum_{i = 0}^m \alpha_i \mid m \in \NN , \alpha_i \in \A  \right\},
\end{equation*}
as well as the truncation
\begin{equation*}
\RA_s \coloneqq \Bigl\{ \sum_{\alpha \in \tA_s} c_{\alpha} x^{\alpha} \: | \: c_{\alpha} \in \R \Bigl\} \quad \text{ and the ring } \quad \RA = \bigcup_{s \in \mathbb{N}} \RA_s.
\end{equation*}

Further, the \emph{degree} of a polynomial $f \in \RA$ with respect to $\A$ is defined as 
\begin{equation*}
\deg_{\A}(f) \coloneqq \min\{ s \in \mathbb{N} \: | \: f \in \RA_s \}.
\end{equation*}
\label{DefinitionRA}
\end{definition}

The set $\A_\infty$ is the \emph{affine semigroup}  generated by $\A$ and the ring $\RA$ is its associated \emph{semigroup algebra}.
Explicitly:
\begin{equation*}
    \RA = \left\{ \sum_{\alpha \in \A_\infty } c_\alpha x^\alpha \mid c_\alpha \in \R \text{ and } c_\alpha \neq 0 \text{ for only finitely many } \alpha \in \A_\infty \right\}.
\end{equation*}
The ring \(\RA\) is closed under addition and multiplication. Multiplicative closure follows immediately because \(\A_\infty\) is closed under addition; thus, the sum of exponents generated by polynomial multiplication remains within \(\A_\infty\).

\begin{example}
Consider the set $\A = \{ e_1,\dots,e_n \} \subset \mathbb{N}^n$ and $\mathcal{B} = \{\pm e_1,\dots, \pm e_n\} \subset \Z^n$, where $e_i$ denotes the standard basis vector.
Then it holds
\begin{equation*}
\RA = \R[x_1,\dots,x_n] \quad \text{ and } \quad \mathbb{R}[\mathcal{B}] = \Lx.
\end{equation*}
\end{example}

\begin{remark}
Observe that, as the exponents in \(\mathcal{A}\) might not be linear independent, \(\deg_\A\) is not necessarily a grading on \(\RA\), that is, does not induce a valuation. Nevertheless, in this work we abuse notation and still call it graduation.
Given \(f,g \in \RA\), we have that,
\begin{itemize}
    \item[i)] \(\deg_{\A}(f \cdot g) \leq \deg_{\A}(f) + \deg_{\A}(g)\),
    \item[ii)] \(\deg_{\A}(f+g) \leq \max\{\deg_{\A}(f),\deg_{\A}(g)\}\).
\end{itemize}
\end{remark}

Algebraically, the Spec of a semigroup algebra $\RA$ is an affine toric variety; see e.g. \cite[Chapter~1]{cox2024toric}.
In order to define its real part, we extend naturally the notion of real radical to $\RA$ and equip the toric varieties with a Harrison topology; see \cite[Exp. 7.1.4]{bochnak2013real}.

\begin{definition}
For a finite set $\A\subset\Z^n$, define the (real) toric variety associated
with $\A$ as
\[
  \mathrm{Sper}(\RA)
  :=
  \bigl\{
  \mathfrak m \subset \RA
  \;:\;
  \mathfrak m \text{ is a maximal real radical ideal}
  \bigr\}.
\]
For $I\subset\RA$, we define its (real) variety as
\[
  \mathcal{V}_{\RA}(I)
  :=
  \{\mathfrak m \in \mathrm{Sper}(\RA) : I \subset \mathfrak m\}.
\]
\end{definition}

Observe that the notation $\mathrm{Sper}(\RA)$ refers to this
real maximal spectrum (maximal ideals with residue field $\R$), rather than to
the full real spectrum of prime cones.

\begin{remark} \label{rmk:homIsSper}
There is a natural identification
\[
  \mathrm{Sper}(\RA)
  =
  \bigl\{
  \lf \in \mathrm{Hom}_{\R\text{-alg}}(\RA,\R)
  \bigr\},
\]
Every unital $\R$-algebra homomorphism $\lf:\RA\to\R$ is
surjective and has a maximal real kernel. Conversely, if $\mathfrak m$ is a
maximal real ideal, then the residue field $\RA/\mathfrak m$ is $\R$, and the
quotient map followed by this identification gives the corresponding
$\R$-algebra homomorphism.
\end{remark}

We next study the coordinate ring $\RA$. Write $\A = \{ \alpha_1,\dots,\alpha_m\}$ and define the surjective map 
\begin{equation} \label{eq:kodairaMap} 
    \psi_\A: \R[z_1,\dots,z_m] \rightarrow \RA, \quad z_i \mapsto x^{\alpha_i}.
\end{equation}
This map allows us to study the ring $\RA$ by studying modulo spaces of $\R[z_1,\dots,z_m]$ as it holds
\[
\R[z_1,\dots,z_m]/\ker(\psi_\A) \simeq \RA.
\]
In particular, it is easy to verify that real radical ideals of $\RA$ correspond to real radical ideals 
of $\R[z_1,\dots,z_m]$ containing $\ker(\psi_\A)$. A first consequence of this is a complete characterization of the points in $\mathrm{Sper}(\RA)$.

\begin{lemma}
\label{lem:Max_ideals_in_RA}
    Let $J \subset \RA$ be a maximal real radical ideal. Then $J = \langle x^\alpha -\beta_\alpha : \alpha \in \A \rangle$ for some $\beta_\alpha \in \R$.
\end{lemma}

\begin{proof}
    Write $\A = \{ \alpha_1,\dots,\alpha_m \}$ and consider the surjective $\R$-algebra homomorphism $\psi_\A$ from \eqref{eq:kodairaMap}. The preimage $H := \psi_\A^{-1}(J)$ of $J$ is a maximal real radical ideal, since $\psi_\A$ is surjective.
    Thus, it can be written in the form $H = \langle z_i - \beta_i : i = 1,\dots,m \rangle$ for some $\beta_i \in \R$. Applying $\psi_\A$ gives the desired result.
\end{proof}

A second trivial consequence of this construction is that the real Nullstellensatz extends to the toric setting.

\begin{proposition}    
    [Toric real Nullstellensatz] \label{prop:toric_nullstellen}
    Let $I$ be an ideal in $\RA$. Then the vanishing ideal of the real variety is the real radical of $I$:
    \[
        \mathcal{I}(\mathcal{V}_{\RA}(I)) = \rIRA.
    \]
\end{proposition}

\begin{proof}
    We shall prove that $\mathcal{I}(\mathcal{V}_{\RA}(I)) \subseteq \rIRA$, as the opposite inclusion is trivial.
    As in the proof of Lemma~\ref{lem:Max_ideals_in_RA}, consider the preimage  $H := \psi_\A^{-1}(I)$ of $I$, where $\psi_\A$ is as defined in \eqref{eq:kodairaMap}.  
    From the previous lemma, we also see that there is an identification of maximal real radical ideals in $\RA$ and the maximal real radical ideals of $\R[z_1,\dots,z_m]$ that contain $\ker(\psi_\A)$.
    Hence, if $f$ vanishes at every $\mathcal{V}_{\RA}(I)$ and $g$
    is any lift of $f$ under $\psi_\A$, then $g$ vanishes on
    $\mathcal{V}_\R(H)$. By the classical real Nullstellensatz,
    $g\in\sqrt[\R]{H}$. Since real radicals commute with quotients under a
    surjective homomorphism,
    $\psi_\A(\sqrt[\R]{H})=\sqrt[\R]{I}$, and therefore
    $f\in\sqrt[\R]{I}$.
\end{proof}

Morally, the previous construction tells us that the real points of $\mathrm{Sper}(\RA)$ are related to real points in $\R^m$ belonging to some toric variety. However, when $\RA \subset \R[x_1, \dots,x_n]$, we might be interested in the real solutions of our systems over $\R^n$.
To understand the relation between these two real spaces, in what follows we explain an alternative way of thinking about (real) toric varieties by considering them as unions of algebraic tori.
For this, we recall the construction of non-normal affine toric varieties associated with an affine semigroup, defined by
$
Y_{\A} := \Spec(\CC[\A]),
$
where $\Spec(\CC[\A])$ denotes the set of maximal ideals in $\CC[\A]$.
We refer the reader to \cite[Appendix~3.A]{cox2024toric} for references and proofs of these results.

Given a cone $\sigma \subseteq \R^n$, we define its dual cone
$
\sigma^\vee := \{ u \in \R^n : (\forall v \in \sigma)\ \langle u,v \rangle \ge 0 \}
$
and its orthogonal cone
$
\sigma^\perp := \{ u \in \sigma^\vee : (\forall v \in \sigma)\ \langle u,v \rangle = 0 \}. 
$
Let $\mathrm{Cone}(A) \subset \R^n$ be the smallest cone containing $\A_\infty$. Since $\A_\infty \subseteq \Z^n$ and it is finitely generated, $\mathrm{Cone}(A)$ is a rational cone, that is, there exists a finite set of generators of the cone that belong to $\Z^n$.
Let $\sigma \subset \R^n$ be the dual cone of $\mathrm{Cone}(A)$, that is,
$
\sigma := \mathrm{Cone}(A)^\vee.
$
For each face $\tau \subseteq \sigma$ (including $\sigma$ itself), we define the \emph{toric orbit}
\begin{align*}
O(\tau) & := 
\left\{
\Lambda \in \Hom_{\Z}(\CC[\A],\CC) :
(\forall \alpha \in \A)\;
\Lambda(x^\alpha) \neq 0 \iff \alpha \in \tau^\perp
\right\} \\ 
& \simeq   
\Hom_{\Z}(\CC[\A \cap \tau^\perp],\CC^*) = \Spec(\CC[\Z(A \cap \tau^\perp)])
\end{align*}
where $\Z(A \cap \tau^\perp)$ is the lattice generated by $A \cap \tau^\perp$.
By computing the Smith normal form of this lattice, we conclude that $O(\tau)$ is isomorphic to
$
\Spec(\CC[x_1^{w_1},\dots,x_r^{w_r}]),
$
where $r$ is the rank of the lattice and $w_1,\dots,w_r$ are its invariant factors (the diagonal entries of the Smith normal form). We have that $r =  
\dim(Y_\A) - \dim(\tau)$.
Moreover, we can think of the points in $\Spec(\CC[x_1^{w_1},\dots,x_r^{w_r}])$ as orbits of points in $(\CC^*)^r$. Namely, we identify
$
p \in \Spec(\CC[x_1^{w_1},\dots,x_r^{w_r}])
$
with
$
\{(q_1,\dots,q_r) \in (\CC^*)^r : (q_1^{w_1},\dots,q_r^{w_r}) = p\};
$
see \cite[Proposition~1.3.18]{cox2024toric}.
In particular, the real points of $\Spec(\CC[x_1^{w_1},\dots,x_r^{w_r}])$ correspond to the orbits
$
\{(q_1,\dots,q_r) \in (\CC^*)^r : (q_1^{w_1},\dots,q_r^{w_r}) \in \R^r\}.
$
Observe that when the index of $\Z(A \cap \tau^\perp)$ in $\Z(\tau^\perp \cap \Z^n)$ (that is, the product of the invariant factors) is odd, we obtain a bijection between the real points of $\Spec(\CC[x_1^{w_1},\dots,x_r^{w_r}])$ and the points of $(\R^*)^r$, given by
\[
p \mapsto \big(\sqrt[w_1]{p_1},\dots,\sqrt[w_r]{p_r}\big) \in (\R^*)^r .
\]

The orbit--cone correspondence \cite[Theorem~3.A.3]{cox2024toric} tells us that
\[
Y_\A = \bigcup_{\tau \subseteq \sigma} O(\tau),
\]
where $\tau$ runs over the faces of $\sigma$.
This correspondence shows that any real point of $Y_\A$, that is, a point in $\mathrm{Sper}(\R[\A])$, can be thought of as a real point in one of its toric orbits, which, as we explained, correspond to special complex points.
When the index of the lattice $\Z(A \cap \tau^\perp)$ in $\Z(\tau^\perp \cap \Z^n)$ is odd, we can identify the real part of $O(\tau)$, that is,
$
O(\tau) \cap \mathrm{Sper}(\R[\A]),
$
with the torus $(\R^*)^r$.
Using this identifications for the cone $\tau = \{0\}$, we can relate the zeros of $f \in \RA$, as a Laurent polynomials, in $\R^*$ with the zeros on $ \mathrm{Sper}(\R[\A])$.

\begin{proposition} \label{prop:solsOnTorus}
Let $f \in \mathbb{R}[\mathcal{A}]$. Suppose there exists a point $p \in (\mathbb{R}^*)^n$ such that $f(p)=0$\footnote{The evaluation is defined by interpreting $f$ as a Laurent polynomial.}. Then there exists a real point $\Lambda \in \mathrm{Sper}(\mathbb{R}[\mathcal{A}])$ such that $\Lambda(f)=0$ and $\Lambda(x^\alpha) \neq 0$ for all $\alpha \in \mathcal{A}$. Moreover, if such a $\Lambda$ exists and the index of $\mathbb{Z}\mathcal{A}$ in $\mathbb{Z}^n$ is odd, then there exists $p \in (\mathbb{R}^*)^n$ such that $f(p)=0$.
\end{proposition}

\begin{proof}
The map which evaluates polynomials at $p\in(\R^*)^n$ is a
homomorphism $\Lambda_p:\R[\A]\to\R$.  Its kernel is a maximal real ideal,
$\Lambda_p(f)=f(p)$, and $\Lambda_p(x^\alpha)=p^\alpha\ne0$ for every
$\alpha\in\A$.  This proves the first claim.

Conversely, suppose that $\Lambda(x^\alpha)\ne0$ for every $\alpha\in\A$.
Multiplicativity makes $\alpha\mapsto\Lambda(x^\alpha)$ a homomorphism from
the semigroup $\NN\A$ to $\R^*$, and hence it extends uniquely to a group
homomorphism $\chi:\Z\A\to\R^*$.  Put the inclusion
$\Z\A\subseteq\Z^n$ in Smith normal form.  If its index is odd, all invariant
factors are odd, so the required roots in $\R^*$ exist and $\chi$ extends to
a character $\widetilde\chi:\Z^n\to\R^*$.  Setting
$p_i:=\widetilde\chi(e_i)$ gives $p\in(\R^*)^n$ and
$p^\alpha=\Lambda(x^\alpha)$ for all $\alpha\in\A$.  Consequently
$f(p)=\Lambda(f)=0$.
\end{proof}


%% file: Sections/Sparse_Extension.tex
\section{Unmixed moment approach}
\label{SectionTruncatedMomentMatricesAndFlatExtension}

In \cite{lasserre2008semidefinite}, Lasserre, Laurent, and Rostalski showed how to use moment matrices to compute the real radical of a zero-dimensional ideal. 
However, the size of these matrices depend on the number of monomials in the algebra at big-enough degree.
In order to shrink these matrices, we propose to use a subset of these monomials given by those belonging to certain affine semigroup.
As we shall show, geometrically, what we are doing is to compute points in the real part of the affine toric variety associated to such semigroup.
Our exposition follows closely the one from  \cite{lasserre2008semidefinite} and our results specialize to theirs when the real part of the affine toric variety is $\R^n$, that is, when $\mathcal{A} = \{0,e_1,\dots,e_n\}$.
For this reason, instead of presenting first their ideas and then rework them to our setting, we follow their structure and present the generalized result directly, making links with \cite{lasserre2008semidefinite} when corresponds.

Before starting, in what follows we present a couple of examples illustrating the new challenges that arise when we consider semigroup algebras instead of the classical polynomial algebra.

\begin{example}[Dimension missmatch]
\label{exp:Dim_not_equal_num_sol}
Consider the monomials
$x^2 y, xy^2
\in \R[x^2 y, x y^2] \subset \R[x,y]$.
The quotient ring $ \R[x,y] / \rr{\langle x^2 y, x y^2 \rangle} $ has Krull dimension $1$, while
$ \R[x^2 y, x y^2] / \rr{\langle x^2 y, x y^2 \rangle} $ has Krull dimension $0$ and defines a unique (real) point, the zero of the toric variety associated to $\R[x^2 y, x y^2]$.
Via the isomorphism from Rmk.~\ref{rmk:homIsSper}, this point correspond to $\Lambda \in \mathrm{Hom}(\R[x^2y,xy^2], \R)$ such that $\Lambda(x^2 y) = \Lambda(x y^2) = 0$.
The discrepancy on the dimension comes from two real lines given by $x = 0$ and $y = 0$ in the variety associated to $\R[x,y] / \rr{\langle x^2 y, x y^2 \rangle} $.

Consider the monomials
$x^2, y^2 \in \R[x^2, y^2] \subset \R[x,y]$.
The quotient ring $ \R[x,y] / \rr{\langle x^2 + y^2 \rangle} $ has Krull dimension $0$ and its associated variety is the point zero in $\R^n$ as $\rr{\langle x^2 + y^2 \rangle} =  \langle x,y\rangle$ over $\R[x,y]$.
However, over $\R[x^2,y^2]$, the ideal $\langle x^2 + y^2 \rangle$ is real radical, so the quotient ring
$\R[x^2,y^2]/\rr{\langle x^2+y^2\rangle}$
has Krull dimension $1$.
We can visualize the lines contained in this variety using the isomorphism
from Rmk.~\ref{rmk:homIsSper}: for every $c\in\R$, there is a point
$\Lambda_c\in\mathrm{Hom}_{\R\text{-alg}}(\R[x^2,y^2],\R)$ such that
$\Lambda_c(x^2)=-\Lambda_c(y^2)=c$. Observe that these points
admit complex lifts $(z,w)$ with $z^2=c$ and $w^2=-c$
in the complex variety associated to
$\R[x,y]/\langle x^2+y^2\rangle$, so real points of the toric variety can
correspond to complex solutions of the original system.

If we modify the previous example and consider the quotient ring $ \R[x^2,y^2] / \rr{\langle x^4 + y^4 \rangle}$, we observe that it corresponds to the point zero given by $\Lambda \in \mathrm{Hom}(\R[x^2,y^2], \R)$, $\Lambda(x^2) = \Lambda(y^2) = 0$, similarly to the classical setting.
\end{example}

Basically, when we look for solutions, not on $\mathbb R^n$, but on the real part of a toric variety, the solutions with zero coordinates might fail to exist in this toric variety (or extra solutions might exist). 
As consequence of Proposition~\ref{prop:solsOnTorus}, these problems arise when we have solutions with zero coordinates, or when we have that the lattice generated by the monomials in $\Z\A$ has even index in $\Z^n$.

\subsection{Sparse moment matrices and operators}
We begin by extending the definition of truncated moment matrices to the sparse setting.

\begin{definition}
\label{DefinitionTruncatedMomentMatrices}
Let $s \in \NN \cup \{\infty\}$ and let $\lf \in \RA_{2s}^* := \{ \phi : \RA_{2s} \to \R \mid \phi \text{ is linear} \}$.
The \emph{truncated moment operator of order $s$} associated with $\lf$, denoted by $M_s(\lf)$, is the operator
\[
  M_s(\lf) := \bigl(\lf(x^\alpha x^\beta)\bigr)_{\alpha,\beta \in \tA_s}.
\]
When $s = \infty$, we drop the subscript from and write directly $M(\lf)$.
\end{definition}
For $s \in \NN \cup \{\infty\}$, we identify polynomials $f\in\RA_s$ with their coefficient sequences
$\vect_s(f) := (f_\alpha)_{\alpha\in\A_s}$.
With this convention, we view $\ker(M_s(\lf))$ as a subspace of $\RA_s$, and we have
\[
  \lf(fg) = \vect_s(f)^\top M_s(\lf)\,\vect_s(g)
  \qquad
  \text{for all } \qquad f,g\in\RA_s.
\] 
Moreover, when $s = \infty$, $\ker(M(\lf))$ is also an ideal, to which we will refer as the \emph{kernel ideal} or $M(\lf)$.

We start with an observation relating the rank of the moment matrix to a
quotient dimension. The following proof follows closely the original one over $\R[x_1,\dots,x_n]$; see {\cite[Lemma~2]{laurent2012approach}}.

\begin{lemma}
\label{lem:rank_equals_dim_cooredinate_ring}
Let $\lf\in\RA^*$ and let $\mathcal{B}$ be a set of monomials.
Then $\mathcal{B}$ indexes a maximal linearly independent set of columns of
$M(\lf)$ if and only if the residue classes of $\mathcal{B}$ form a basis of
$\RA/\ker( M(\lf))$.
In particular,
\[
  \rk(M(\lf)) = \dim_\R\bigl(\RA/\ker(M(\lf))\bigr).
\]
\end{lemma}

\begin{proof}
A finite linear combination among columns indexed by $\B$ is exactly a
polynomial $f$ in the linear span of $\B$ satisfying $M(\lf)\vect_\infty(f)=0$, that is,
$f\in\ker(M(\lf))$. Thus these columns are linearly independent if and only
if the residue classes of $\B$ are linearly independent in the quotient.
Likewise, a column indexed by a monomial $x^\alpha$ lies in the span of the
columns indexed by $\B$ if and only if the residue class of $x^\alpha$ lies in
the span of the residue classes of $\B$. The two maximality statements are
therefore equivalent, and the rank--dimension identity follows.
\end{proof}

Lemma \ref{lem:rank_equals_dim_cooredinate_ring} implies that $\ker(M(\lf))$ is
zero-dimensional, i.e., $\dim(\RA/\ker (M(\lf)))<\infty$, precisely when
$M(\lf)$ has finite rank.

\begin{remark}[Points and rank-one functionals]\label{rmk:rank1andPoints}
    Observe that $\rk(M(\lf)) = 1$ if and only if $\ker(M(\lf))$ is a maximal real radical ideal.
    Further, there is a one-to-one correspondence between rank-one functionals $\lf \in \RA^*$ such that  $\lf(1) = 1$  and maximal real radical ideals.
    This follows from Lemma~\ref{lem:Max_ideals_in_RA} as, if $J \subset \RA$ is a maximal real radical ideal, we have $J = \langle x^\alpha -\beta_\alpha : \alpha \in \A \rangle$ for some $\beta_\alpha \in \R$. Hence, we can consider $\lf_J \in \RA^*$  such that $\lf_J(x^\alpha) = \beta_\alpha$.
    Because of this, in what follows we identify  points and rank-one functionals.
\end{remark}

Computing the real roots of a sparse polynomial system can be approached via the
real radical ideal generated by the defining polynomials.
For the characterization of real radicals we consider \emph{positive linear forms}.

\begin{definition}
\label{Def:positive_linear_forms}
Let $s \in \NN \cup \{\infty\}$ and let $\lf \in \RA_{2s}^*$.
Define the quadratic form associated with $\lf$ by
\[
  Q_\lf^s : \RA_s \to \R, \qquad f \mapsto \lf(f^2).
\]
Its kernel is the linear subspace given by
\[
  \ker(Q_\lf^s)
  :=
  \bigl\{ f \in \RA_s : \lf(fg)=0 \text{ for all } g\in\RA_s \bigr\}.
\]
We call $\lf$ \emph{positive} if $\lf(f^2)\ge 0$ for all $f\in\RA_s$.
\end{definition}

By Definition \ref{DefinitionTruncatedMomentMatrices}, $\lf$ is positive if and
only if $M_s(\lf)$ is positive semidefinite, denoted $M_s(\lf)\succeq 0$.
Using the identification with polynomials $f$ and $\mathrm{vec}(f)$, we identify $\ker(M(\lf))$ with $\ker(Q_\lf^\infty)$. 

\begin{example}
    Consider $\lf_{\mathfrak m} \in \mathrm{Sper}(\RA)$. Then $\lf_{\mathfrak m}$ is positive, since $ 0 \leq \lf_{\mathfrak m}(f)^2 = \lf_{\mathfrak m}(f^2)$. 
\end{example}

Positivity forces kernel ideals to be real radical.
The following lemma is the direct analogue over $\RA$ of
\cite[Proposition~3.6]{lasserre2008semidefinite}; we include the proof to make
the adaptation explicit.

\begin{lemma}
\label{LemmaKernelQF}
Let $\lf\in\RA^*$.
If $\lf$ is positive, then the ideal $\ker(M(\lf))$ is real radical.
\end{lemma}

\begin{proof}
First note that $\ker(M(\lf))$ is an ideal.
Assume $\lf$ is positive and let $\sum_{i=1}^m p_i^2 \in \ker(M(\lf))$ for some
$p_i\in\RA$.
We show $p_i \in \ker(M(\lf))$ for every $i$.
For any $g\in\RA$ we have
\[
  0
  =
  \lf\!\left(g^2 \sum_{i=1}^m p_i^2\right)
  =
  \sum_{i=1}^m \lf(g^2 p_i^2).
\]
Since each term $\lf(g^2 p_i^2)\ge 0$ by positivity, it follows that
$\lf(g^2 p_i^2)=0$ for all $i$ and all $g$.
It remains to show that $\lf(q^2)=0$ implies $\lf(q)=0$.
Fix $q\in\RA$ with $\lf(q^2)=0$.
For any $t\in\R$,
\[
  0 \le \lf((q+t)^2)
  = \lf(q^2) + 2t\,\lf(q) + t^2 \lf(1)
  = 2t\,\lf(q) + t^2 \lf(1).
\]
This inequality for all $t\in\R$ forces $\lf(q)=0$.
Hence, if we consider $q=g \ p_i$, we get
$\lf(g \ p_i)=0$ for every $g\in\RA$, and so $p_i\in\ker(M(\lf))$.
\end{proof}

Our next example shows that the converse of the previous lemma does not hold, namely, there the kernel ideal might be real radical even when $\Lambda$ is not positive.

\begin{example}
\label{ExampleKernelML}
Let $\A=\{(2,0),(0,1)\}$ and define $\lf\in\RA^*$ by
$\lf(1)=1$, $\lf(x_1^2)=1$, and $\lf(x^\alpha)=0$ for all
$\alpha\in\A_\infty\setminus\{(0,0),(2,0)\}$.
Then $\ker(M(\lf))$ is the ideal generated by $x_2$ and $1-x_1^2$, hence
$\ker(M(\lf))\subset \RA$ is real radical.
However, $\lf$ is not positive since
\[
  \lf\bigl((1-x_1^2)^2\bigr)
  =
  \lf(1) - 2\lf(x_1^2) + \lf(x_1^4)
  =
  -1 < 0.
\]
\end{example}

Next, we generalize basic combinatorial properties of conic combinations of
elements in $\mathrm{Sper}(\RA)$.

\begin{lemma}
\label{LemmaEvaluationsAndKernels}
Let $\mathfrak m\subset\RA$ be a maximal real radical ideal and let
$\lf_{\mathfrak m}$ be the corresponding element of $\mathrm{Sper}(\RA)$.
Then $M(\lf_{\mathfrak m})$ has rank $1$ and
$\ker(M(\lf_{\mathfrak m}))=\mathfrak m$.
More generally, let $\lf = \sum_{i=1}^r \lambda_i \lf_i$, where $\lambda_i\neq 0$ and $\lf_i\in\mathrm{Sper}(\RA)$ are pairwise distinct. Then $\rk(M(\lf))=r$.
If, in addition, all $\lambda_i>0$, then $\ker(M(\lf)) = \bigcap_{i=1}^r \ker(\lf_i)$.
\end{lemma}

\begin{proof}
It is immediate that $M(\lf_{\mathfrak m})$ has rank $1$.
Moreover,
\[
  \mathfrak m \subseteq \ker(M(\lf_{\mathfrak m})) \subsetneq \RA,
\]
and maximality of $\mathfrak m$ implies $\ker(M(\lf_{\mathfrak m}))=\mathfrak m$.

For the second claim, put $\mathfrak m_i=\ker(\lf_i)$. Distinct
maximal ideals are pairwise comaximal, so the Chinese remainder theorem gives
$p_1,\dots,p_r\in\RA$ with $\lf_j(p_i)=\delta_{ij}$. In the bases induced by
these interpolation polynomials, the bilinear form
$(f,g)\mapsto\lf(fg)$ has diagonal matrix
$\operatorname{diag}(\lambda_1,\dots,\lambda_r)$. Since every $\lambda_i$ is
nonzero, this form has rank $r$, and so does $M(\lf)$.
Further, if each $M(\lf_i)\succeq 0$ and
$\lf=\sum_{i=1}^r \lambda_i\lf_i$ is a conic
combination of the $\lf_i$ with $\lambda_i>0$, we obtain
\[
  \ker(M(\lf))=\bigcap_{i=1}^r \ker(M(\lf_i))=\bigcap_{i=1}^r \ker(\lf_i).
\]
The equality follows form the identity 
\[
  \ker(\lambda A+\mu B) = \ker(A)\cap\ker(B),
\]
 for $\lambda,\mu >0$ and positive semidefinite matrices $A,B$.
\end{proof}

\begin{example}[Continuation of Example \ref{exp:Dim_not_equal_num_sol}.]
Let $\Lambda = \lf_{\mathfrak m_1} + \lf_{\mathfrak m_2}$, where
$\mathfrak m_1 = \langle x^{\alpha_1},\; x^{\alpha_2} - 1 \rangle$ and
$\mathfrak m_2 = \langle x^{\alpha_2},\; x^{\alpha_1} - 1 \rangle$.
Then
\[
  M(\Lambda) =
  \begin{bmatrix}
    2 & 1 & 1 & 1 & 0 & 1 & \cdots \\
    1 & 1 & 0 & 1 & 0 & 0 & \cdots \\
    1 & 0 & 1 & 0 & 0 & 1 & \cdots \\
    1 & 1 & 0 & 1 & 0 & 0 & \cdots \\
    0 & 0 & 0 & 0 & 0 & 0 & \cdots \\
    1 & 0 & 1 & 0 & 0 & 1 & \cdots \\
    \vdots & \vdots & \vdots & \vdots & \vdots & \vdots & \ddots
  \end{bmatrix},
\]
where the columns are indexed by
$1,\ x^{\alpha_1},\ x^{\alpha_2},\ x^{2\alpha_1},\ x^{\alpha_1 + \alpha_2},\ x^{2\alpha_2},\dots$.
In particular, $\rk(M(\Lambda))=2$.
\end{example}

The following lemma characterize zero-dimensional radical ideals in $\RA$.

\begin{lemma}
\label{lem:Real_Nullstellensatz_Toric_variety}
Let $I\subset\RA$ be a zero-dimensional real radical ideal, i.e.,
$|\mathcal{V}_{\RA}(I)|<\infty$.
Then
\[
  \dim(\RA/I) = |\mathcal{V}_{\RA}(I)|
  \qquad\text{and}\qquad
  \RA/I \cong \prod_{i=1}^{\dim(\RA/I)} \R.
\]
\end{lemma}

\begin{proof}
Since $I$ is zero-dimensional, $\RA/I$ is Artinian; see \cite[p.~92]{atiyah2018introduction}.
Hence $\RA/I$ is finite-dimensional as an $\R$-vector space; write
$r:=\dim(\RA/I)$.
Moreover, $\RA/I$ is reduced because $I$ is real radical.
By \cite[Theorem~8.7]{atiyah2018introduction}, we obtain an isomorphism
\[
  \RA/I \cong K_1 \times \dots \times K_l,
\]
where each $K_i$ is a finite field extension of $\R$.
Since $\R\subseteq \RA/I$, each $K_i$ is a finite-dimensional field extension of
$\R$, hence $K_i\in\{\R,\mathbb{C}\}$.
No factor can be $\mathbb C$: in such a factor the two nonzero
elements $1$ and $i$ have squares summing to zero. Multiplying them by the
idempotent supported on that factor would contradict the real-reduced
property of $\RA/I$.
Thus $K_i=\R$ for all $i$, and therefore $l=r$ and
$\RA/I\cong \R^r$.

Finally, the maximal ideals of $\RA$ containing $I$ correspond to the coordinate
projections of $\R^r$, hence $|\mathcal{V}_{\RA}(I)|=r=\dim(\RA/I)$.
\end{proof}

The next theorem shows that the elements of $\mathrm{Sper}(\RA)$ are the
building blocks of $\RA^*$, generalizing \cite[Theorem 9]{laurent2012approach}.

\begin{theorem}
\label{TheoremKernMRealRadicalIffLinearCombination}
\label{TheoremFiniteRankMomentMatrix}
Let $\lf\in\RA^*$ be such that $\rk(M(\lf))=r<\infty$.
Then the following are equivalent:
\begin{enumerate}
  \item[i)] $\ker(M(\lf))$ is real radical.
  \item[ii)] $\lf = \sum_{i=1}^r \lambda_i \lf_i$, where
  $\lambda_i\in\R\setminus\{0\}$ and
  $\lf_i\in\mathrm{Sper}(\RA)$ are pairwise distinct.
\end{enumerate}
Furthermore, if $\lf$ is positive, then $\lambda_i>0$ and $\mathcal{V}_{\RA}(\ker(M(\lf))) = \{\lf_1,\dots,\lf_r\}$.
\end{theorem}

\begin{proof}
Put $J:=\ker(M(\lf))$. If $J$ is real radical, then
Lemma~\ref{lem:rank_equals_dim_cooredinate_ring} and
Lemma~\ref{lem:Real_Nullstellensatz_Toric_variety} give an isomorphism
$\Theta:\RA/J\to\R^r$. The functional induced by $\lf$ on $\RA/J$ is
therefore of the form
\[
 \overline\lf=\sum_{i=1}^r\lambda_i\pi_i,
\]
where $\pi_i$ is the $i$-th coordinate projection. Every $\lambda_i$ is
nonzero, since the bilinear form induced by $M(\lf)$ on $\RA/J$ is
nondegenerate and has diagonal matrix
$\operatorname{diag}(\lambda_1,\dots,\lambda_r)$ in the standard basis.
Pulling the coordinate projections back through $\Theta$ gives (ii).

Conversely, suppose that (ii) holds and put
$\mathfrak m_i:=\ker(\lf_i)$. Then
$\bigcap_i\mathfrak m_i\subseteq J$. The Chinese remainder theorem and
Lemma~\ref{lem:rank_equals_dim_cooredinate_ring} give
\[
 \dim\bigl(\RA/\textstyle\bigcap_i\mathfrak m_i\bigr)=r
 =\dim(\RA/J),
\]
so $J=\bigcap_i\mathfrak m_i$ and is real radical.

Finally, if $\lf$ is positive, then
$\lambda_i=\overline\lf(e_i)=\overline\lf(e_i^2)>0$. The last claim follows
from Lemma~\ref{LemmaEvaluationsAndKernels}.
\end{proof}

\subsection{A semidefinite description of real radicals and toric varieties}
\label{subsubsec:SOS_characterization_toric_varieties}

We combine the previous results to obtain a semidefinite characterization of
real radical ideals in $\RA$ via positive linear forms. Concretely, we study
positive $\lf\in\RA^*$ satisfying $\ker(M(\lf))=\sqrt[\R]{I}$.
All arguments in this subsection are direct adaptations of the corresponding ones in the classical moment
matrix approach in \cite{lasserre2008semidefinite}: they use only algebraic manipulations and positivity and do not
depend on a specific grading. Hence, when the proof is identical, we omit the proof and cite the original source.

Fix an ideal $I\subset\RA$ and denote by $\sqrt[\R]{I}$ its real radical
in $\RA$. Consider the convex set
\begin{equation}
\label{eq:K_I}
\K \coloneqq
\Bigl\{
\lf\in\RA^*
\ \Big|\
\lf(1)=1,\ M(\lf)\succeq 0,\ \lf(p)=0\ \text{for all }p\in I
\Bigr\}.
\end{equation}
For any $\lf\in\K$, the ideal $\ker(M(\lf))$ is real radical; see
Lemma~\ref{LemmaKernelQF}. Moreover, $\ker(M(\lf))$ contains $I$, and hence also
contains $\rIRA$.

Assume from now on that $I$ is zero-dimensional in the toric sense, i.e.,
$|\mathcal{V}_{\RA}(I)|<\infty$.
Then Lemma~\ref{lem:Real_Nullstellensatz_Toric_variety} implies that
$\dim(\RA/\rIRA)=|\mathcal{V}_{\RA}(I)|$.
Moreover, Lemma~\ref{lem:rank_equals_dim_cooredinate_ring} implies the following rank bound for 
$\lf\in\K$,
\[
  \rk(M(\lf))
  =
  \dim\bigl(\RA/\ker(M(\lf))\bigr)
  \le
  \dim(\RA/\rIRA)
  =
  |\mathcal{V}_{\RA}(I)|.
\]
Additionally, there exist elements of $\K$ attaining this maximum rank. For example,
if $\mathcal{V}_{\RA}(I)=\{\mathfrak m_1,\dots,\mathfrak m_r\}$ and
$\lf_{\mathfrak m_i}\in\mathrm{Sper}(\RA)$ are the corresponding evaluation
maps, then
\[
  \lf
  \coloneqq
  \frac{1}{r}\sum_{i=1}^r \lf_{\mathfrak m_i}
  \in \K
\]
has rank equal to $r$.

\begin{definition}
\label{DefinitionGenericLF}
Let $\K$ be as in \eqref{eq:K_I} and assume $|\mathcal{V}_{\RA}(I)|<\infty$.
We say that $\lf\in\K$ is \emph{generic} if $M(\lf)$ has maximal rank, i.e.,
\[
  \rk(M(\lf)) = |\mathcal{V}_{\RA}(I)|.
\]
\end{definition}
Generic elements of $\K$ recover the real radical $\rIRA$. This relies on a
standard kernel property for positive semidefinite moment matrices.

\begin{lemma}[{\cite[Lemma~4]{laurent2012approach}}]
\label{LemmaGenericLF}
Assume $|\mathcal{V}_{\RA}(I)|<\infty$.
Then $\lf\in\K$ is generic if and only if
\[
  \ker(M(\lf)) \subseteq \ker(M(\lf'))
  \quad \text{for all }\lf'\in\K.
\]
Moreover, for every generic $\lf\in\K$ one has $\ker(M(\lf)) = \rIRA$.
\end{lemma}

In particular, $\K$ is a polytope whose vertices are precisely the evaluation
maps at the points of the toric variety. This observation will play an essential role for the computation for only one root in Section \ref{sec:One_root}.

\begin{theorem}
\label{thm:K_is_polytope}
Assume $|\mathcal{V}_{\RA}(I)|<\infty$.
Then
\[
  \K
  =
  \mathrm{conv}\Bigl(\left\{\lf_{\mathfrak m} : \mathfrak m \in \mathcal{V}_{\RA}(I)\right\}\Bigr).
\]
Equivalently, $\K$ is a polytope with vertex set
$\left\{\lf_{\mathfrak m} : \mathfrak m \in \mathcal{V}_{\RA}(I)\right\}$.
\end{theorem}

\begin{proof}
The inclusion ``$\supseteq$'' is immediate: each $\lf_{\mathfrak m}$ satisfies
$\lf_{\mathfrak m}(1)=1$, is positive, and vanishes on $I$ whenever
$I\subseteq \mathfrak m$; hence every convex combination lies in $\K$.

For the reverse inclusion, let $\lf\in\K$.
Since $I\subseteq\ker(M(\lf))$ and the latter ideal is real radical
by Lemma~\ref{LemmaKernelQF}, we have
$\rIRA\subseteq\ker(M(\lf))$.
Since $M(\lf)\succeq 0$ and $\rk(M(\lf))<\infty$,
Theorem~\ref{TheoremFiniteRankMomentMatrix} implies that $\lf$ is a conic
combination of finitely many elements of $\mathrm{Sper}(\RA)$. Its
support is $\mathcal V_{\RA}(\ker (M(\lf)))$, which is contained in
$\mathcal{V}_{\RA}(I)$.
Normalizing by $\lf(1)=1$ yields a convex combination, hence
$\lf\in \mathrm{conv}(\{\lf_{\mathfrak m} : \mathfrak m \in \mathcal{V}_{\RA}(I)\})$.
\end{proof}

\begin{example} (Continuation of Example \ref{exp:Dim_not_equal_num_sol}.)
\label{ExampleDefinitionf1f2f3}
Let $\A=\{(2,1),(1,2)\}=\{\alpha_1,\alpha_2\}$ and consider
\begin{align*}
f_1 & = x^{2 \alpha_1 + \alpha_2},\\
f_2 & = x^{2 \alpha_1 + \alpha_2} - x^{\alpha_2} + x^{\alpha_1} + 1,\\
f_3 & = x^{2\alpha_1} - x^{\alpha_1},\\
f_4 & = x^{2\alpha_2} - x^{\alpha_2} + x^{2 \alpha_1 + 4 \alpha_2},
\end{align*}
and the ideal $J:=\langle f_1,f_2,f_3,f_4\rangle$.
Let $\lf=\frac{1}{2}(\lf_{\mathfrak m_1}+\lf_{\mathfrak m_2})$.
Then $\lf\in\K$ is generic, and therefore
\[
\sqrt[\R]{J}   = \ker(M(\lf)) = I = \bigl\langle
  x^{\alpha_1 + \alpha_2},\;
  x^{\alpha_1 + \alpha_2} - x^{\alpha_1} - x^{\alpha_2} + 1
  \bigr\rangle.
\]
\end{example}

\subsection{Sparse flat extension theorem}
Even though the polytope $\K$ is finite dimensional, its ambient space is infinite dimensional.
To recover the real solutions of the system, we would like to work on finite-dimensional spaces, so we will truncate the moment matrices and $\K$.
In what follows, we focus on truncated moment matrices and prove that the classical flat extension theorem \cite{curto1998flat} holds in the sparse setting. 
For this, we adapt the construction of
\cite[Section~2]{laurent2009generalized} to the toric coordinate
ring $\RA$. The main difference with our construction is that we need
to verify that the multiplication maps respect the relations between
the monomials in $\RA$.

Let $\A=\{\alpha_1,\dots,\alpha_m\}$. Consider
the map $\psi_\A : \R[z_1,\dots,z_m] \rightarrow \RA$ defined in \eqref{eq:kodairaMap} which maps $z_i \mapsto \psi_\A(z_i) := x^{\alpha_i}$.
Consider the kernel of this map
$T_\A:=\ker(\psi_\A) \subseteq \R[z_1,\dots,z_r]$ and fix a finite generating set $G_\A$ of $T_\A$. We
define
\begin{align}\label{eq:rhoA}
 \rho_\A:=\max\{\deg(g):g\in G_\A\},
\end{align}
with $\rho_\A=0$ if $T_\A= \langle 0 \rangle$. 
\begin{theorem}[Sparse flat extension]
\label{TheoremSparseFlatExtension}
Let $s\geq\rho_\A$ and let $\lf\in\RA_{2s}^*$ satisfy
\[
 M_s(\lf)\succeq0,
 \qquad
 \rk (M_s(\lf))=\rk (M_{s-1}(\lf))=:r.
\]
Then there is a unique $\widetilde\lf\in\RA^*$ such that
\[
 \widetilde\lf|_{\RA_{2s}}=\lf,
 \qquad M(\widetilde\lf)\succeq0,
 \qquad \rk (M(\widetilde\lf))=r.
\]
\end{theorem}

\begin{proof}
Let $K_s:=\ker(M_s(\lf))$. As $\rk(M_s(\lf))=\rk(M_{s-1}(\lf))$, we can choose a monomial column basis
$\B\subseteq\RA_{s-1}$ of $M_s(\lf)$.
We consider the $\R$-vector space $V$ spanned by the monomials in $\B$.
Every element in $\RA_s$ has a projection to $V$ given by the map be $\pi:\RA_s\longrightarrow V$ which sends each $p\in\RA_s$ to its unique projection $\pi(p)\in V$ characterized by
\begin{equation}
  p-\pi(p)\in K_s.                                  
  \label{eq:direct-reduction}
\end{equation}
Indeed, the residue classes of the monomials in $\B$ form a basis of $\RA_s/K_s$: they are linearly independent by the choice of $\B$, and both spaces have dimension $r$. Thus $\RA_s=V\oplus K_s$.
Observe that, if $h\in K_s\cap\RA_{s-1}$, then
\begin{equation}
 x^{\alpha_i} \ h\in K_s.                                        
 \label{eq:direct-truncated-ideal}
\end{equation}
Indeed, if $q=\pi(q)+p\in\RA_s$ with $p\in K_s$, then
$\lf((x^{\alpha_i}  \ h)\ q)
 =\lf\bigl(h \ x^{\alpha_i} \ \pi(q)\bigr)+\lf\bigl((x^{\alpha_i} \ h) \ p\bigr)=0.$ 
Here the first term vanishes because $h\in K_s$ and $x^{\alpha_i} \, \pi(q)\in\RA_s$, while the second vanishes because $p \in K_s$ and $x^{\alpha_i} \, h\in\RA_s$.

Given $p,q\in V$, consider the bilinear form
$\langle p,q\rangle_\lf:=\lf(pq)$. This is an inner product on $V$.
In fact, its Gram matrix is the principal submatrix of $M_s(\lf)$
indexed by $\B$.
This inner product is nondegenerate and its Gram matrix is positive definite because
$M_s(\lf)\succeq0$ and the corresponding columns are linearly independent.

For $i=1,\dots,m$, we define the multiplication maps
\begin{equation}
  X_i:V\longrightarrow V,\qquad X_i (p) :=\pi(x^{\alpha_i} \, p).       
  \label{eq:direct-multiplication}
\end{equation}
With respect to the inner product induced by $\lf$, these maps are self-adjoint,
\[
 \langle X_i (p),q\rangle_\lf
 =\lf(\pi(x^{\alpha_i} \, p) \ q) = \lf(x^{\alpha_i} \, p \, q)
 =\langle p,X_i(q)\rangle_\lf, \]
They also commute.  For $p,q\in V$,
\begin{align*}
 \langle X_i(X_j(p)),q\rangle_\lf
 =\langle X_j(p),X_i(q)\rangle_\lf
   =\lf(x^{\alpha_j} \, p \ x^{\alpha_i} \, q)
 =\lf(x^{\alpha_i} \, p \ x^{\alpha_j} \, q)
   =\langle X_j(X_i(p)),q\rangle_\lf.
\end{align*}
The nondegeneracy of the inner product gives $X_i \circ X_j=X_j \circ X_i$.

Let $e:=\pi(1)\in V$.  For $k \leq s$ and $w \in \{1\dots m\}^k$, define $x^w= \prod_{j = 1}^k x^{\alpha_{w_j}}$ and $X^w=X_{w_1} \circ \cdots \circ X_{w_k}$.
We claim that
\begin{equation}
  X^w(e)=\pi(x^w).
  \label{eq:direct-word-reduction}
\end{equation}
This can be proved by induction on $k$. The case $k=0$ is immediate. Inductively, if $x^{(w_1,\dots,w_{k-1})}-X^{(w_1,\dots,w_{k-1})}(e) \in K_s\cap\RA_{s-1}$, by \eqref{eq:direct-truncated-ideal}, we have that
$x^{\alpha_{w_k}} (x^{(w_1,\dots,w_{k-1})}-X^{(w_1,\dots,w_{k-1})}(e)) \in K_s$.
Consequently, commutativity of the multiplication maps completes the induction as 
\begin{multline*}
    \pi(x^{(w_1,\ldots,w_k)})= \pi(x^{\alpha_{w_k}} \ x^{(w_1,\dots,w_{k-1})}) = \pi(x^{\alpha_{w_k}} \ X^{(w_1,\ldots,w_{k-1})}(e)) \\ =X_{w_k} \circ X^{(w_1,\ldots,w_{k-1})}(e) = X^{(w_1,\dots,w_{k})}(e).
\end{multline*}

The multiplication maps induce an algebra homomorphism $\Phi : \R[x_1,\dots,x_m] \rightarrow \operatorname{End}_\R(V)$ sending $x_i \mapsto \Phi(x_i) := X_i$.
Since $\RA\simeq\R[x_1,\dots,x_m]/T_\A$, it remains to prove that the multiplication maps respect the relations in $T_\A=\ker(\psi_\A)$.
Recall $G_\A$ is a generating set of $T_\A$.
Fix $g\in G_\A$ and $b \in \B$.
By \eqref{eq:direct-word-reduction}, there is $k\leq s-1$ and $w(b)\in \{1\dots m\}^k$ such that $b=X^{w(b)}(e)$.
Moreover, every monomial of $g$ has degree at most
$\rho_\A\leq s$, so by
\eqref{eq:direct-word-reduction} and $\psi_\A(g)=0$ we have that
\begin{multline*}
  \Phi(g) (b) = g(X_1,\dots,X_m) (b) = g(X_1,\dots,X_m) \circ X^{w(b)}(e) \\ =
   X^{w(b)} \circ g(X_1,\dots,X_m) (e) = 
  X^{w(b)}(\pi(\psi_\A(g))) = 0.
\end{multline*}
Thus $g(X_1,\dots,X_m)=0$ on $V$.  This holds for all $g\in G_\A$, and hence
for the whole toric ideal $T_\A$ as wanted.
Therefore, $\Phi$ factors through the quotient by $T_\A=\ker(\psi_\A)$. 
We abuse notation and consider $\Phi : \RA\simeq\R[x_1,\dots,x_m]/T_\A \rightarrow \operatorname{End}_\R(V)$. Observe that the image of $\Phi$ at any monomial in $\RA$ does not depend on the chosen factorization of that monomial.

We now define the flat extension $\widetilde\lf \in \RA^*$  as \[\widetilde\lf(f):=\lf\bigl(\Phi(f)(e)\bigr).\]
First observe that, as $1-e\in K_s$, for every $v\in V$,
$\lf(v)=\langle v,e\rangle_\lf.$
Moreover, since the maps $X_i$ commute and are self-adjoint, $\Phi(f)$ is self-adjoint for every $f\in\RA$. It follows that, for all $f,g\in\RA$,
\begin{equation}\label{eq:direct-full-gram}
    \widetilde\lf(fg) = \lf(\langle\Phi(fg)(e)) 
    =\bigl\langle\Phi(fg)(e),e\bigr\rangle_\lf
    =\bigl\langle\Phi(f)(e),\Phi(g)(e)\bigr\rangle_\lf.
\end{equation}
We now show that $\widetilde\lf$ extends $\lf$. 
Let $x^\tau\in\RA_{2s}$ be a monomial. We can write $x^\tau=x^\alpha x^\beta$ with $x^\alpha,x^\beta\in\RA_s$. By \eqref{eq:direct-word-reduction},
$\Phi(x^\alpha)(e)=\pi(x^\alpha)$ and $\Phi(x^\beta)(e)=\pi(x^\beta)$.
Therefore,
\[
    \widetilde\lf(x^\tau)
    =\bigl\langle\pi(x^\alpha),\pi(x^\beta)\bigr\rangle_\lf
    =\lf(x^\alpha x^\beta)
    =\lf(x^\tau),
\]
where the middle equality follows because $x^\alpha-\pi(x^\alpha)$ and $x^\beta-\pi(x^\beta)$ belong to $K_s$. By linearity, $\widetilde\lf|_{\RA_{2s}}=\lf$.

Equation \eqref{eq:direct-full-gram} shows that the full moment matrix $M(\widetilde\lf)$ is the Gram matrix of the vectors $\{\Phi(x^\alpha)(e):x^\alpha\in\RA\}$ in $V$ with respect to the inner product $\langle \bullet,\bullet \rangle_\lf$. Hence $M(\widetilde\lf)\succeq0$ and $\rk (M(\widetilde\lf))\leq\dim V=r$. On the other hand, if $b\in\B$, then \eqref{eq:direct-word-reduction} gives $\Phi(b)(e)=b$. Since $\B$ is a basis of $V$, these vectors are linearly independent, and therefore $\rk (M(\widetilde\lf))=r$. Thus, $\widetilde\lf$ is a flat extension of $\lf$.

For uniqueness, let $\widehat\lf\in\RA^*$ be another positive semidefinite extension of $\lf$ with $\rk(M(\widehat\lf))=r$. The columns indexed by $\B$ remain linearly independent and hence form a basis of the column space of $M(\widehat\lf)$. For $b\in\B$ and $i=1,\ldots,m$, set $q_{i,b}:=x^{\alpha_i} \, b-X_i(b)\in K_s$. Since $\widehat\lf$ extends $\lf$ and is positive semidefinite,
\[
 \widehat\lf(q_{i,b}^2)=\lf(q_{i,b}^2)=0
 \quad\Longrightarrow\quad q_{i,b}\in\ker(M(\widehat\lf)).
\]
The same argument gives $1-e\in\ker(M(\widehat\lf))$. Thus, in the basis $\B$, multiplication by $x^{\alpha_i}$ is represented by $X_i$ and the column of $1$ is represented by $e$. Consequently, the column of every monomial $x^\alpha$ is represented by $\Phi(x^\alpha)(e)$, and
\[
 \widehat\lf(x^\alpha)=\bigl\langle\Phi(x^\alpha)(e),e\bigr\rangle_\lf
 =\widetilde\lf(x^\alpha).
\]
By linearity, $\widehat\lf=\widetilde\lf$.
\end{proof}

Morally, the main difference with the previous flat extension theorems is the requirement that condition $s\geq\rho_\A$, which guarantees that
the generators of $T_\A$ are contained in the truncation under consideration.
As our next example shows, this condition can not be avoided.

\begin{example}
Take $\RA=\R[x^2,x^3]$.  For the map $\psi_\A$ from
\eqref{eq:kodairaMap}, we have
$T_\A=\langle x_1^3-x_2^2\rangle$.  At $s=1$, set
\[
 \lf(1)=1,\quad \lf(x^2)=1,\quad \lf(x^3)=2,\quad
 \lf(x^4)=1,\quad \lf(x^5)=2,\quad \lf(x^6)=4.
\]
Then
\[
 M_1(\lf)=
 \begin{bmatrix}1&1&2\\1&1&2\\2&2&4\end{bmatrix}\succeq0,
 \qquad \rk (M_1(\lf))=\rk (M_0(\lf))=1.
\]
Nevertheless, a rank-one full extension would be a character of $\RA$ and
would have to respect $(x^2)^3=(x^3)^2$, whereas the prescribed values give
$1^3\neq2^2$.  The obstruction is invisible at order one because its defining
binomial has degree three.  This is precisely the obstruction removed by
$s\geq\rho_\A$.
\end{example}

The following argument is the toric analogue of
\cite[Proposition~3.6]{lasserre2008semidefinite}.

\begin{corollary}
\label{CorollaryDirectFlatKernel}
Following the notation from Theorem~\ref{TheoremSparseFlatExtension}, we have that
\[
 \ker (M(\widetilde\lf))=\langle\ker (M_s(\lf))\rangle.
\]
Moreover, every monomial column basis
$\B\subseteq\RA_{s-1}$ of $M_s(\lf)$ is a basis of
$\RA/\ker(M(\widetilde\lf))$.
\end{corollary}

\begin{proof}
In what follows, we use the same notations and constructions from the proof of Theorem~\ref{TheoremSparseFlatExtension}. Recall $K_s:=\ker(M_s(\lf))$ and let $J:=\langle K_s\rangle$. For $p\in K_s$, we have $\Phi(p)(e)=\pi(p)=0$. The vectors $\{ \Phi(f)(e) : f\in\RA\}$ span $V$ because they include the basis $\B$. Since the maps in the image of $\Phi$ commute, $\Phi(p)(\Phi(f)(e))=\Phi(f)(\Phi(p)(e))=0$ for every $f$, and therefore $J\subseteq\ker(\Phi)$.

Conversely, the relations $1-e\in K_s$ and
\[
 x^{\alpha_i} \, b-X_i(b)=x^{\alpha_i} \, b-\pi(x^{\alpha_i} \, b)\in K_s\qquad(b\in\B)
\]
reduce every monomial modulo $J$ to an element of $V$. Hence $\RA=V+J$. If $v\in V\cap J$, then $v=\Phi(v)(e)=0$, so $\RA=V\oplus J$ and $\ker(\Phi)=J$. Finally, by \eqref{eq:direct-full-gram} and the same spanning argument,
\[
 f\in\ker(M(\widetilde\lf))
 \iff \Phi(f)(e)=0
 \iff \Phi(f)=0.
\]
Thus $\ker(M(\widetilde\lf))=J$, and the residue classes of $\B$ form a basis of $\RA/J$.
\end{proof}

\begin{corollary}
\label{CorollaryDirectFlatAtomic}
The extension $\widetilde{\Lambda} \in \RA^*$ from Theorem~\ref{TheoremSparseFlatExtension} can be decomposed as 
\[
 \widetilde\lf=\sum_{k=1}^r\lambda_k\lf_{\mathfrak m_k},
 \qquad \lambda_k>0,
\]
where $\mathfrak m_1,\dots,\mathfrak m_r$ are distinct elements of
$\mathrm{Sper}(\RA)$. In particular,
$\ker(M(\widetilde\lf))$ is real radical.
\end{corollary}

\begin{proof}
By Theorem~\ref{TheoremSparseFlatExtension}, $\widetilde\lf$ is
positive and $\rk(M(\widetilde\lf))=r$. Hence
$\ker(M(\widetilde\lf))$ is real radical by
Lemma~\ref{LemmaKernelQF}, and the result follows directly from
Theorem~\ref{TheoremFiniteRankMomentMatrix}.
\end{proof}

\begin{remark}[Relation with connected-to-one flat extensions]
In \cite{laurent2009generalized}, Laurent and Mourrain consider a finite
monomial set connected to $1$ in a polynomial ring.  This means that every
nonconstant monomial in the set can be obtained from $1$ by successive
multiplication by variables, while remaining inside the set.  Their flat
extension theorem is formulated in terms of the border of such a monomial set.

The construction above also uses multiplication maps, but these maps
must respect the relations in $T_\A$ in order to be well-defined on
$\RA\simeq\R[x_1,\dots,x_m]/T_\A$.  This is the reason for the additional
condition $s\geq\rho_\A$.  Moreover, since $\A$ is allowed to be a subset of
$\Z^n$, the coordinate ring $\RA$ may contain Laurent monomials with negative
exponents.  Thus, the two approaches are related, but in the toric setting one
must additionally verify the compatibility with the relations in $T_\A$.
\end{remark}

\begin{example}[Continuation of Example \ref{ExampleDefinitionf1f2f3}.]
Consider the linear form $\Lambda=\lf_{\mathfrak m_1}+\lf_{\mathfrak m_2}$,
where
$\mathfrak m_1=\langle x^{\alpha_1},\; x^{\alpha_2}-1\rangle$ and
$\mathfrak m_2=\langle x^{\alpha_2},\; x^{\alpha_1}-1\rangle$.
Then
\[  
  M_2(\Lambda)= 
  \begin{array}{c | c c c c c c}
  & 1&  x^{\alpha_1}&  x^{\alpha_2}&  x^{2\alpha_1}&  x^{\alpha_1 + \alpha_2}&  x^{2\alpha_2} \\  \hline
  1 &  2 & 1 & 1 & 1 & 0 & 1 \\
  x^{\alpha_1} &  1 & 1 & 0 & 1 & 0 & 0 \\
  x^{\alpha_2}&   1 & 0 & 1 & 0 & 0 & 1 \\
  x^{2\alpha_1}&  1 & 1 & 0 & 1 & 0 & 0 \\
  x^{\alpha_1 + \alpha_2}&   0 & 0 & 0 & 0 & 0 & 0 \\
  x^{2\alpha_2} &   1 & 0 & 1 & 0 & 0 & 1
  \end{array},
\]
In this case, $\rho_\A = 0$ as the monomials $x^\alpha_1$ and $x^\alpha_2$ are algebraically independent.
Moreover,
$\rk(M_1(\Lambda))=\rk(M_2(\Lambda))=2$, so the flat extension condition holds
and we can extend $\Lambda$ to an infinite moment matrix $M(\widetilde{\Lambda})$.
Concretely, we can express higher monomials modulo the ideal
$\langle \ker(M(\widetilde{\Lambda}))\rangle=\langle \ker(M_2(\Lambda))\rangle$.
For instance,
\[
  x^{3\alpha_1}
  =
  \bigl(x^{2\alpha_1} - x^{\alpha_1}\bigr)
  (x^{\alpha_1}+1) + x^{\alpha_1}.
\]
Since $x^{2\alpha_1} - x^{\alpha_1} \in \ker(M(\widetilde{\Lambda}))$, the column of
$M(\widetilde{\Lambda})$ indexed by $x^{3\alpha_1}$ coincides with the column
indexed by $x^{\alpha_1}$.
Iterating this reduction yields all columns of $M(\widetilde{\Lambda})$ from the
finite block $M_2(\Lambda)$.
\end{example}

\subsection{Computing solutions on toric varieties}
\label{subsec:computing_toric_varieties}

We explain how the algorithm of \cite{laurent2012approach} extends to solving sparse
polynomial systems over the real part of the toric variety associated with $\RA$.
The main idea is to compute the real radical ideal $\rIRA$ from truncated moment
data. A central ingredient is the characterization of real radical ideals via
positive linear forms (Theorem
\ref{TheoremKernMRealRadicalIffLinearCombination}).

To make the method effective, we consider finite-dimensional
truncations $\RA_s\subset \RA$.
Fix an ideal $I=\langle f_1,\dots,f_\mu\rangle \subset \RA$
generated by the polynomials $f_i\in\RA$ we want to solve.
For $t\in\NN$, define the set of \emph{prolongations up to degree $t$} by
\begin{equation}
\label{DefintionOfHt}
\Ht
\coloneqq
\Bigl\{
x^\alpha \,f_i 
\ \Big|\
\alpha\in\A_\infty,\ i=1,\dots,\mu,\ 
\deg_{\A}(x^\alpha) \le t-\deg_{\A}(f_i)
\Bigr\}.
\end{equation}
We also define the truncated analogue of $\K$,
\begin{equation}
\label{DefinitionOfKt}
\K_t
\coloneqq
\Bigl\{
\lf \in \RA_t^*
\ \Big|\
\lf(1)=1,\ 
M_{\lfloor t/2 \rfloor}(\lf)\succeq 0,\ 
\lf(f)=0\ \text{for all } f\in \Ht
\Bigr\}.
\end{equation}

A key observation is the relation between
generic elements of $\K$ and the real radical $\rIRA$ (Lemma
\ref{LemmaGenericLF}). The next lemma provides the truncated analogue.

\begin{lemma}
\label{LemmaGenericLFinKt}
Let $\lf\in\K_t$. The following are equivalent:
\begin{enumerate}
\item[(i)] $M_{\lfloor t/2\rfloor}(\lf)$ has maximum rank among the moment
matrices associated with elements of $\K_t$;
\item[(ii)] $\ker(M_{\lfloor t/2\rfloor}(\lf)) \subseteq
\ker(M_{\lfloor t/2\rfloor}(\lf'))$ for every $\lf'\in\K_t$;
\item[(iii)] $\lf\in\operatorname{relint}(\K_t)$.
\end{enumerate}
\end{lemma}

\begin{proof}
Write $M=M_{\lfloor t/2\rfloor}(\lf)$ and
$N=M_{\lfloor t/2\rfloor}(\lf')$. Since both are positive semidefinite,
\[
 \ker(M+N)=\ker(M)\cap\ker(N).
\]
If $M$ has maximum rank, then $(M+N)/2 \in \K_t$ cannot have larger
rank, so its kernel must equal $\ker (M)$; hence $\ker (M)\subseteq\ker (N)$.
The converse follows immediately from the rank--nullity theorem.

Assume (ii). 
$N$ vanishes on $\ker(M)$, because $\ker (M)\subseteq\ker (N)$. Thus, for sufficiently small
$\varepsilon>0$,
\[
 (1+\varepsilon)M-\varepsilon N\succeq0.
\]
All affine constraints defining $\K_t$ are preserved by the corresponding
functional $\lf+\varepsilon(\lf-\lf')$. The standard segment
characterization of relative interior now gives (iii). Conversely, if
$\lf$ is in the relative interior, the extension $\lf+\varepsilon_{\lf'}(\lf-\lf')$ is contained in $\K_t$ for 
every $\lf'\in\K_t$ and small enough $\varepsilon_{\lf'}$. For $u\in\ker (M)$, positivity of
$(1+\varepsilon)M-\varepsilon N$ gives
$0\le-\varepsilon u^TNu$, hence $Nu=0$. Therefore (ii) holds.
\end{proof}

The set $\K_t$ is a spectrahedron \cite{Spectrahedron}.
This is crucial computationally: standard semidefinite programming (SDP) solvers
can approximate points in $\K_t$.

Observe that by Lemma~\ref{LemmaGenericLFinKt}, for any $\Lambda \in \K_t$ in the relative interior of $\K_t$, we obtain the same vector space $\ker\bigl(M_{\lfloor t/2\rfloor}(\lf)\bigr)$. This motivates the following definition.

\begin{definition}
\label{DefinitionGenericLFandNt}
An element satisfying Lemma~\ref{LemmaGenericLFinKt} is called
\emph{generic}. For a generic $\lf\in\K_t$, set
\[
 \mathcal N_t:=\ker (M_{\lfloor t/2\rfloor}(\lf)) \subseteq\RA.
\]
The space $\mathcal N_t$ is independent of the chosen generic element.
\end{definition}

Generic elements can be obtained by choosing a point in the relative
interior of the spectrahedron. Such an interior point can be computed
via a constant-objective SDP describing the feasible set, that is,
\begin{align}
\mathrm{argmin}_{\lf \in \K_t} 1.
\label{AlignSDP}
\end{align}
In practice, we might use an interior-point procedure to
obtain a maximum-rank solution.

The following lemma shows that any polynomial in $\sqrt[\R]{I}$ is
eventually in the kernel of every linear functional in $\K_t$, when $t$
is big enough.

\begin{lemma}
\label{lem:eventual-radical-kernel}
Let $F\subseteq\sqrt[\R]{I}$ be finite. There is $t_0$ such that, for every
$t\ge t_0$ and $s \in \NN$, we have that
\[
 \{ x^\alpha \ f : f \in F, x^\alpha \in \RA_s\} \subseteq\ker (M_{\lfloor t/2\rfloor + s}(\lf)) \qquad (\forall \lf\in\K_{t+2s}).
\]
\end{lemma}

\begin{proof}
If $F=\emptyset$, the result is immediate, so assume that
$F\ne\emptyset$. Fix $s\in\NN$ and $\lf\in\K_{t+2s}$, and put
$r=\lfloor t/2\rfloor$. We first record the truncated kernel rule used
below. If $p\in\ker (M_r(\lf))$ and
$\deg_\A(p)+\deg_\A(q)\le r-1$, then
$pq\in\ker (M_r(\lf))$.
To prove this directly, first let $q=x^{\alpha_i}$. Since
$x^{2\alpha_i}p\in\RA_r$, we have
\[
 \lf\bigl((x^{\alpha_i}p)^2\bigr)
 =\lf\bigl(p\,x^{2\alpha_i}p\bigr)=0.
\]
Positivity of $M_r(\lf)$ gives
$x^{\alpha_i}p\in\ker (M_r(\lf))$. Iterating this argument proves the
claim when $q$ is a monomial, and the general case follows by linearity.

Fix $f\in F$. By Lemma~\ref{lem:Real_Nullstellensatz_Toric_variety} there are
$m\ge1$ and fixed
polynomials $q_j,u_i\in\RA$ such that
\[
 f^{2m}+\sum_jq_j^2=\sum_i u_i f_i.
\]
For all sufficiently large $t$, the prolongation constraints imply
$f_i\in\ker (M_r(\lf))$: every product of $f_i$ with a row monomial of
$M_r(\lf)$ belongs to $\mathcal H_t\subseteq\mathcal H_{t+2s}$.
Increasing $t$ once more, the truncated
kernel rule gives $u_i f_i\in\ker (M_r(\lf))$. Hence the left-hand side of the
display lies in that kernel. Evaluating it at the row $1$ and using positivity
shows $\lf(f^{2m})=0$, so $f^m\in\ker (M_r(\lf))$. If $m>1$, put
$k=\lceil m/2\rceil$. Since $f^m$ is in the kernel, testing it against
$f^{2k-m}$ gives $\lf(f^{2k})=0$, and positivity gives
$f^k\in\ker (M_r(\lf))$. Repeating this descent yields
$f\in\ker (M_r(\lf))$. All polynomials used here are fixed, so one common
$t_0$ works for the finite set $F$.
Enlarge this $t_0$, if necessary, so that
$\left\lfloor\frac{t_0}{2}\right\rfloor
 \ge 1+\max_{f\in F}\deg_\A(f).$
Put $R:=\lfloor(t+2s)/2\rfloor=r+s$. For every $f\in F$, the inclusion
$f\in\ker (M_r(\lf))$ gives $\lf(f^2)=0$. Since
$M_R(\lf)\succeq0$, it follows that $f\in\ker (M_R(\lf))$. Moreover, for
every monomial $x^\alpha\in\RA_s$,
\[
 \deg_\A(f)+\deg_\A(x^\alpha)
 \le \left\lfloor\frac{t_0}{2}\right\rfloor-1+s
 \le R-1.
\]
Applying the truncated kernel rule at order $R$ gives
$x^\alpha f\in\ker (M_R(\lf))$, which is the desired conclusion.
\end{proof}

The next theorem shows that, when $\lf$ is generic, the sequence $\{\Nt\}_t$ leads to a (convergent) sequence of ideals approximating the real radical of $I$.
Its eventual equality statement is the toric analogue of the
convergence result \cite[Proposition~4.6]{lasserre2008semidefinite}.

\begin{theorem}
\label{TheoremNtEqRealRadicalI}
The spaces $\mathcal N_t$ satisfy
\[
 \mathcal N_0\subseteq\mathcal N_1\subseteq\cdots
 \subseteq\sqrt[\R]{I}.
\]
Consequently the ideals $\langle\mathcal N_t\rangle$ form an increasing
filtration of $\sqrt[\R]{I}$, and
\[
 \langle\mathcal N_t\rangle=\sqrt[\R]{I}
\]
for all sufficiently large $t$.
\end{theorem}

\begin{proof}
Let $\lf\in\K_{t+1}$ be generic and restrict it to $\RA_t$. The maximum-kernel
property at truncation $t$ gives
\[
 \mathcal N_t\subseteq
 \ker (M_{\lfloor t/2\rfloor}(\lf|_{\RA_t})) \subseteq \mathcal N_{t+1}.
\]

For every $\mathfrak m\in\mathcal V_{\RA}(I)$, the truncated evaluation
$\lf_{\mathfrak m}|_{\RA_t}$ belongs to $\K_t$. Lemma
\ref{LemmaGenericLFinKt} therefore gives
\[
 \mathcal N_t\subseteq\bigcap_{\mathfrak m\in\mathcal V_{\RA}(I)}
 \ker(\lf_{\mathfrak m})=\sqrt[\R]{I}.
\]
Finally choose a finite generating set $F$ of $\sqrt[\R]{I}$ and apply
Lemma~\ref{lem:eventual-radical-kernel}. For large $t$, every element of $F$
belongs to the kernel of every feasible moment matrix, hence in particular to
$\mathcal N_t$. Thus $\sqrt[\R]{I}\subseteq\langle\mathcal N_t\rangle$, and
the reverse inclusion was just proved.
\end{proof}

In order to make the previous theorem into an algorithm, our next
theorem introduces a sufficient condition to test when the previous
sequence of ideals converged.
We will use this condition as a stopping criterion to detect when we
had computed the real radical. This sufficient condition is based on the flat
extension theorem.
The two rank conditions below are the toric counterparts of those
in \cite[Proposition~4.4]{lasserre2008semidefinite}.

\begin{theorem}
\label{TheoremNtEqualsRealRadicalWithRankCondition}
Let $D=\max_i\deg_\A(f_i)$ and $d=\lceil D/2\rceil$. Let
$\lf\in\K_t$ and suppose that either
\begin{enumerate}
\item[(i)] $\rk (M_s(\lf))=\rk (M_{s-1}(\lf))$ for some
$\max\{D,\rho_\A\}\le s\le\lfloor t/2\rfloor$, or
\item[(ii)] $\rk (M_s(\lf))=\rk(M_{s-d}(\lf))$ for some
$\max\{d,\rho_\A\}\le s\le\lfloor t/2\rfloor$.
\end{enumerate}
Then
\[
 \sqrt[\R]{I}\subseteq J:=\langle\ker (M_s(\lf)) \rangle.
\]
If $\lf$ is generic in $\K_t$, then $J=\sqrt[\R]{I}$, and every monomial set
indexing a column basis of $M_{s-1}(\lf)$ is a basis of
$\RA/\sqrt[\R]{I}$.
\end{theorem}

\begin{proof}
Since the ranks of the moment matrices are nondecreasing, either condition
implies
$\rk (M_s(\lf))=\rk (M_{s-1}(\lf))$.
Thus, Theorem~\ref{TheoremSparseFlatExtension} and
Corollaries~\ref{CorollaryDirectFlatKernel} and
\ref{CorollaryDirectFlatAtomic} show that $J$ is real radical and that the
flat extension of $\lf$ has the form
\[
 \widetilde\lf=\sum_{j=1}^r\lambda_j\lf_j,
 \qquad \lambda_j>0,
 \qquad \mathcal V_{\RA}(J)=\{\lf_1,\dots,\lf_r\}.
\]

Under (i), for every $q\in\RA_s$ we have
$\deg_\A(qf_i)\le s+D\le2s\le t$, so the prolongation constraints give
$f_i\in\ker (M_s(\lf))$. Hence $I\subseteq J$ and
$\sqrt[\R]{I}\subseteq J$.

Under (ii), choose a monomial column basis
$\B\subseteq\RA_{s-d}$. By Corollary~\ref{CorollaryDirectFlatKernel},
$\B$ is a basis of $\RA/J$, so there are interpolation polynomials
$p_1,\dots,p_r\in\lin_\R(\B)$ such that $\lf_k(p_j)=\delta_{jk}$. Since
\[
 \deg_\A(p_j^2f_i)\le2(s-d)+D\le2s\le t,
\]
the prolongation constraints give
$0=\lf(p_j^2f_i)=\lambda_j\lf_j(f_i)$. Hence
$\mathcal V_{\RA}(J)\subseteq\mathcal V_{\RA}(I)$, and therefore
$\sqrt[\R]{I}\subseteq J$.

If $\lf$ is generic, Lemma~\ref{LemmaGenericLFinKt}, applied to the truncated
evaluations at the points of $\mathcal V_{\RA}(I)$, gives
$\ker (M_s(\lf))\subseteq\sqrt[\R]{I}$. Hence
$J=\sqrt[\R]{I}$, and the last statement follows from
Corollary~\ref{CorollaryDirectFlatKernel}.
\end{proof}

\begin{remark}
Theorem \ref{TheoremNtEqualsRealRadicalWithRankCondition} gives a way to compute a sparse border basis connected-to-one of $\RA/\rIRA$ (Definition~\ref{def:connected-to-one}).
Indeed, once the rank condition holds, any monomial set indexing a column basis
of $M_{s-1}(\lf)$ yields a basis of $\RA/\rIRA$.
By selecting such columns greedily in increasing $\deg_{\A}$ order, one can
construct a border basis connected-to-one. For more details, see the Appendix.
\end{remark}

Putting all together, we obtain
Algorithm~\ref{AlgorithmMomentApproachForLx} which computes the real
solutions of a sparse system, when they are finite. Our algorithm uses
the border basis to approximate the solutions, as detailed in the
Appendix.
This algorithm is an
extension of the one proposed in
\cite[Section~4.4]{lasserre2008semidefinite} to the toric setting, and
specialise to it when $\A$ is the canonical basis of $\Z^n$.

\begin{algorithm}
\caption{Computing roots on toric varieties}
\label{AlgorithmMomentApproachForLx}
\textbf{Input:} A finite $\A\subset\Z^n$ and
$f_1,\dots,f_\mu\in\RA$ with finite $\mathcal V_{\RA}(f_1,\dots,f_\mu)$.
\begin{itemize}
\item[(1)] Compute a generating set of $T_\A$, its bound $\rho_\A$, and set
$D=\max_i\deg_\A(f_i)$, $d=\lceil D/2\rceil$, and
$t=2\max\{D,\rho_\A\}$.
\item[(2)] If $\K_t=\emptyset$, report that there is no solution. Otherwise
compute a generic element $\lf\in\K_t$.
\item[(3)] Test conditions (i) and (ii) of
Theorem~\ref{TheoremNtEqualsRealRadicalWithRankCondition}.
\item[(4)] If neither holds, replace $t$ by $t+1$ and return to step (2).
\item[(5)] Otherwise set $J=\langle\ker (M_s(\lf))\rangle$, compute a monomial
basis of $\RA/J$, and recover all points by the sparse eigenvalue method in the
Appendix.
\end{itemize}
\textbf{Output:} Generators of $\sqrt[\R]{I}$, a basis of
$\RA/\sqrt[\R]{I}$, and $\mathcal V_{\RA}(I)$.
\end{algorithm}

In order to prove the correctness of
Algorithm~\ref{AlgorithmMomentApproachForLx}, the next lemma shows
that the sufficient condition of
Theorem~\ref{TheoremNtEqualsRealRadicalWithRankCondition} is
eventually satisfied.
This lemma adapts the existence and eventual-rank statements of
\cite[Proposition~4.6]{lasserre2008semidefinite} to the semigroup algebra.

\begin{lemma}
\label{TheoremCheckingExistenceOfRoots}
Let $I\subseteq\RA$ be an ideal.
Put $D=\max_i\deg_\A(f_i)$ for the fixed generators
$I=\langle f_1,\dots,f_\mu\rangle$.
\begin{enumerate}
\item[(i)] $\mathcal V_{\RA}(I)=\emptyset$ if and only if $\K_t=\emptyset$
for all sufficiently large $t$.
\item[(ii)] If $1\le|\mathcal V_{\RA}(I)|<\infty$, then there are fixed
integers $s_0$ and $t_0$ such that, for every $t\ge t_0$,
every $s \geq s_0$ and every
$\lf\in\K_{t+2s}$,
\[
 \rk (M_{s}(\lf))=\rk (M_{s-1}(\lf)),
 \qquad s \ge\max\{D,\rho_\A\}.
\]
\end{enumerate}
\end{lemma}

\begin{proof}
If the variety is empty, Lemma~\ref{lem:Real_Nullstellensatz_Toric_variety} gives a certificate
$1+\sum_jq_j^2=\sum_i u_i f_i$. For large $t$ all terms on the right are
prolongation constraints. Applying any $\lf\in\K_t$ gives
$1+\sum_j\lf(q_j^2)=0$, contradicting positivity. Conversely, a point of the
variety gives a truncated linear form in every $\K_t$.

Assume now that the variety is finite and nonempty. Choose a
monomial basis $\B$ of $\RA/\sqrt[\R]{I}$, put
$c:=\max_{x^\beta\in\B}\deg_\A(x^\beta)$,
and choose $s_0\ge\max\{D,\rho_\A,c+1\}$. For each monomial
$x^\alpha\in\RA_{s_0}$, choose $r_\alpha\in\lin_\R(\B)$ such that
$k_\alpha:=x^\alpha-r_\alpha\in\sqrt[\R]{I}$.
Let $t_0$ be given by Lemma
\ref{lem:eventual-radical-kernel} for the set $F := \{k_\alpha : x^\alpha\in\RA_{s_0}\}$.
Fix $t\ge t_0$, $s\ge s_0$, and $\lf\in\K_{t+2s}$. Apply Lemma
\ref{lem:eventual-radical-kernel}
with $t$ replaced by $t+2s_0$ and $s$ replaced by $s-s_0$.
Since $(t+2s_0)+2(s-s_0)=t+2s$, this gives
\[
 x^\delta k_\alpha\in
\ker\bigl(M_{\lfloor(t+2s_0)/2\rfloor+s-s_0}(\lf)\bigr)
 =\ker\bigl(M_{\lfloor t/2\rfloor+s}(\lf)\bigr)
\]
for every monomial
$x^\delta\in\RA_{s-s_0}$ and every $k_\alpha\in F$.
By splitting a factorization, every monomial $x^\gamma\in\RA_s$ can be
written as $x^\gamma=x^\delta x^\alpha$ with
$x^\delta\in\RA_{s-s_0}$ and
$x^\alpha\in\RA_{s_0}$. Hence its column in
$M_s(\lf)$ equals the column of $x^\delta r_\alpha$, because
$x^\delta k_\alpha\in\ker\bigl(M_s(\lf)\bigr)$. The latter belongs to
$\RA_{s-s_0+c}\subseteq\RA_{s-1}$. Thus every
column of $M_s(\lf)$ is in the
column space of $M_{s-1}(\lf)$, and the reverse inclusion is immediate.
\end{proof}

\begin{theorem}
\label{TheoremAlgoLxTerminates}
Algorithm~\ref{AlgorithmMomentApproachForLx} terminates after finitely many
steps and is correct.
\end{theorem}

\begin{proof}
Lemma~\ref{TheoremCheckingExistenceOfRoots} shows that, at a sufficiently high
order, either $\K_t$ is empty or condition (i) in
Theorem~\ref{TheoremNtEqualsRealRadicalWithRankCondition} holds. Thus the
loop terminates. At termination the chosen functional is generic, so that
the same theorem identifies $J$ with $\sqrt[\R]{I}$. The quotient basis and
the sparse eigenvalue theorem then recover precisely the points of
$\mathcal V_{\RA}(I)$.
\end{proof}

%% file: Sections/ComputingOneRoot.tex
\section{Computing One Root}
\label{sec:One_root}
Algorithm~\ref{AlgorithmMomentApproachForLx} relies on computing a generic point in $\K_t$. In this section we explore what happens when we pick extreme points of this spectrahedron. We shall show that this procedure allow us to compute a unique real point of the sparse system.

As we saw in Theorem \ref{thm:K_is_polytope}, the spectrahedron $\K$ is, in fact, a polytope, and its vertices correspond to the solutions of the original system $h_1=\dots=h_m=0$.
Analogously, we shall see that, for $t$ sufficiently large, the truncated feasible set $\K_t$ can be projected to a polytope. 
To state this precisely, define the truncation (projection) map
\begin{equation*}
  \pi_s : \RA^* \to \RA_s^*, \qquad \lf \mapsto \lf|_{\RA_s}.
\end{equation*}

\begin{theorem}[Projection to polytope]
\label{thm:KtIsPolytope}
Consider $I := \langle  f_1,\dots, f_\mu \rangle \subset \RA$ such that $1\le|\mathcal V_{\RA}(I)|<\infty$. Let $\K$ be as in \eqref{eq:K_I} and 
$\K_t$ as in \eqref{DefinitionOfKt}.
There are integers $s_0$ and $t_0$
such that,
for every $s\ge s_0$ and every $t\ge t_0+2s$,
one has
\[
 \pi_{2s}(\K_t)=\pi_{2s}(\K).
\]
This common projection is a simplex of dimension $\dim\pi_{2s}(\K_t)=|\mathcal V_{\RA}(I)|-1$ whose vertices are the truncated evaluation functionals
\[
\{\Lambda_{\mathfrak m}|_{\RA_{2s}} : \mathfrak m \in \mathcal V_{\RA}(I) \}.
\]
\end{theorem}

\begin{proof}
Take $s_0$ and $t_0$ as in
Lemma~\ref{TheoremCheckingExistenceOfRoots}(ii). Fix $s\ge s_0$,
$t\ge t_0+2s$, and $\lf\in\K_t$. Since
 $t-2s\ge t_0$
 and $\K_t=\K_{(t-2s)+2s}$,
Lemma~\ref{TheoremCheckingExistenceOfRoots}(ii) implies
\[
 \rk\bigl(M_s(\lf)\bigr)=\rk\bigl(M_{s-1}(\lf)\bigr).
\]
As $s\ge\max(\max_i \deg(f_i),\rho_\A)$, Theorem~\ref{TheoremSparseFlatExtension}
extends $\lf|_{\RA_{2s}}$ to a positive flat functional
$\widehat\lf\in\RA^*$.
Theorem~\ref{TheoremNtEqualsRealRadicalWithRankCondition}(i) and
Corollary~\ref{CorollaryDirectFlatKernel} give directly
$I\subseteq\sqrt[\R]{I}\subseteq\ker\bigl(M(\widehat\lf)\bigr)$.
Hence $\widehat\lf\in\K$. Since $\widehat\lf$ extends
$\lf|_{\RA_{2s}}$, this proves
$\pi_{2s}(\K_t)\subseteq\pi_{2s}(\K)$. The reverse inclusion is immediate.

Finally, the proof of Lemma~\ref{TheoremCheckingExistenceOfRoots} provides a
monomial basis of $\RA/\sqrt[\R]{I}$ contained in $\RA_{s_0-1}$. Therefore,
for $s\ge s_0$, the truncated evaluations on $\RA_{2s}$ are affinely
independent. The last claim follows from
Theorem~\ref{thm:K_is_polytope}.
\end{proof}

In general, $\K_t$ need not be a polytope or even a polyhedron.

\begin{example}
\label{exp:K_t_not_polyhedron}
Let $\A=\{e_1\}$, consider $h_1=(x-1) \ x \ (x+1)$, and truncation degree $t=3$.
Then $\RA=\R[x]$ and
\[
  \K=\mathrm{conv}\left(\lf_{1},\lf_{0},\lf_{-1}\right).
\]
Moreover $\mathcal{H}_3=\{h_1\}$, hence
\begin{align*}
\K_3
&= \Bigl\{\lf\in \R[x]_3^* \ \Big|\ \lf(1)=1,\ M_1(\lf)\succeq 0,\ \lf(h_1)=0\Bigr\} \\
&= \Bigl\{\lf\in \R[x]_3^* \ \Big|\ \lf(1)=1,\ \lf(x^2)\ge \lf(x)^2,\ \lf(x^3)=\lf(x)\Bigr\}.
\end{align*}
This shows that $\K_3$ is not even a polyhedron: it is a convex set with infinitely many supporting hyperplanes.
Moreover, the inclusion $\pi_3(\K)\subsetneq \K_3$ is strict.

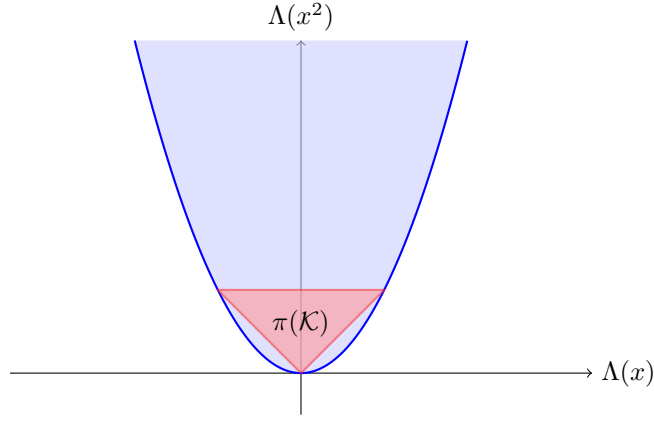
\begin{figure}[H]
\begin{center}
    \begin{tikzpicture}[scale=1.1]

    \draw[->] (-3.5,0) -- (3.5,0) node[right] {$\lf(x)$};
    \draw[->] (0,-0.5) -- (0,4) node[above] {$\lf(x^2)$};

    \def\xmin{-2}
    \def\xmax{2}

    \begin{scope}
        \clip (\xmin,0) rectangle (\xmax,4); 
        \fill[blue!20,opacity=0.6]
            plot[domain=\xmin:\xmax,samples=200]
                (\x,{(\x)^2})
            -- (\xmax,4)
            -- (\xmin,4)
            -- cycle;
    \end{scope}

    \draw[thick,blue]
        plot[domain=\xmin:\xmax,samples=200]
            (\x,{(\x)^2});

    \coordinate (A) at (-1,1);
    \coordinate (B) at (0,0);
    \coordinate (C) at (1,1);

    \filldraw[fill=red!40,opacity=0.6,draw=red!80,thick]
        (A) -- (B) -- (C) -- cycle;

    \node at (0,0.6) {$\pi(\mathcal{K})$};
\end{tikzpicture}
\caption{The set $\K_3$ is drawn in blue and the projection of $\K$ is drawn in red.}
\end{center}
\end{figure}

\end{example}

Solutions at infinity may even cause $\K_t$ not to be a polyhedron
at any sufficiently high truncation order.

\begin{example}  
For $\A=\{e_1,e_2,e_3\} \subset \NN^3$,  we have $\RA=\R[x,y,z]$ is
the classical polynomial ring. Consider
\[
 {f_1=x^3+x,\qquad f_2=x^3+x+y,\qquad f_3=x^3+x+z.}
\]
Their homogenizations with respect to a new variable $w$ vanish on the
projective line $x=w=0$.

The real variety of this system consists only of the origin. We
claim that $\K_t$ is not a polyhedron for any $t\ge3$. Put
$d=\lfloor t/2\rfloor$ and let $L_t\subseteq\RA_t^*$ be the affine subspace
of the functionals $\lf$ such that $\lf(1)=1$ and which vanish on every
nonconstant monomial except those of degree $2d$ in $y,z$.

For $\lf\in L_t$, we have that $\lf(f) = 0$ for every $f \in \Ht$ as
no polynomial in this vector space contain the monomials $1$ or any
monomial of degree $2d$ only involving $y$ and $z$.
Moreover, up to zero rows and columns,
\[
 M_d(\lf)=1\oplus H_d(\lf),
 \qquad \text{where }
 H_d(\lf)=
 \bigl(\lf(y^{2d-i-j}z^{i+j})\bigr)_{0\le i,j\le d}.
\]
Hence $L_t\cap\K_t$ is affinely isomorphic to the cone of positive
semidefinite Hankel matrices of size $d+1$. This cone is not
polyhedral: it is a classical result that it
has infinitely many distinct extreme rays given by the positive
semidefinite Hankel matrices of rank one,
\[
 \R_{\ge0} \cdot (V_d(u,v) \, V_d(u,v)^{\mathsf T}),
 \qquad
\text{where } V_d(u,v) := (u^d,u^{d-1}v,\dots,v^d)^{\mathsf T},
 \qquad \text{and } [u:v]\in\mathbb P^1(\R).
\]
Since the intersection of a polyhedron with an affine subspace is a
polyhedron, $\K_t$ is not a polyhedron for any $t\ge3$.
\end{example}


\subsection{Randomized computation of one real solution}

In Algorithm~\ref{AlgorithmMomentApproachForLx},
to compute a generic $\lf\in\K_t$, we solve the optimization problem
\begin{equation*}
\mathrm{argmin}_{\lf\in\K_t}\lf(1).
\end{equation*}
and choose a generic minimizer.
However, we could replace the constant objective $1\in\RA_t$
by a polynomial $c\in\RA_{2s}\subseteq
\RA_t$, where $s\geq1$ and $t\geq2s$:
\begin{equation*}
  \mathrm{argmin}_{\lf\in\K_t}\lf(c).
\end{equation*}
If $t$ is sufficiently large and $c\in\RA_{2s}$ is generic, then the
minimum is attained at a vertex of $\pi_{2s}(\K_t) =\pi_{2s}(\K)$,
hence at the restriction of an evaluation at a real solution.
In this case, for $t$ sufficiently large, the optimizer has a rank-one
moment matrix whose kernel corresponds to the real solution associated
to this vertex.
Via Lemma \ref{lem:Max_ideals_in_RA}, we can recover the coordinates of this solution.

In what follows, we present a randomized variant of
Algorithm~\ref{AlgorithmMomentApproachForLx} which allow us to compute
only one real root.
As we shall show, every real root has a positive probability of being
chosen and this probability depends on certain solid angles.
For this, we fix $s \ge 1$, choose a basis of $\RA_{2s}$, and use the
Euclidean inner product for which this basis is orthonormal. We choose
the objective $c$ uniformly on the unit sphere of
$\RA_{2s}$. Equivalently, we may take independent standard Gaussian
coefficients, since normalization does not change the minimizer. The
optimization variable is still a functional in $\RA_t^*$, for
any $t \ge 2s$.

\begin{lemma}
\label{lem:quadratic-objectives-expose-roots}
Assume that $\mathcal V_{\RA}(I)$ is finite. For every $s \ge 1$, each truncated
evaluation $\lf_{\mathfrak m}|_{\RA_{2s}}$ is a vertex of the polytope
$\pi_{2s}(\K)$.
\end{lemma}

\begin{proof}
For $\mathfrak m\in\mathcal V_{\RA}(I)$, consider
$
 q_{\mathfrak m}:=
 \sum_{\alpha_i\in\A}
 \left(x^{\alpha_i}-\lf_{\mathfrak m}(x^{\alpha_i})\right)^2\in\RA_2.
 $
 By Theorem~\ref{thm:K_is_polytope}, $\pi_{2s}(\K) =
\operatorname{conv}\left\{
 \lf_{\mathfrak n}|_{\RA_{2s}}:\mathfrak n\in\mathcal V_{\RA}(I)
 \right\}.
$
Then $\lf_{\mathfrak m}(q_{\mathfrak m})=0$, whereas
$\lf_{\mathfrak n}(q_{\mathfrak m})>0$ for every
$\mathfrak n\ne\mathfrak m$, since the algebra generators separate distinct
points. Thus $q_{\mathfrak m}$ uniquely minimizes at
$\lf_{\mathfrak m}|_{\RA_{2s}}$ because $\RA_2\subseteq\RA_{2s}$.
\end{proof}

From now on, we fix a norm for the space of polynomial
$\RA_{2s}$.
For each solution $\mathfrak m\in\mathcal V_{\RA}(I)$, we define its
normal cone by
\[
 N_{2s}(\mathfrak m):=
 \left\{c\in\RA_{2s}:\lf_{\mathfrak m}(c)
 \le\lf_{\mathfrak n}(c)\text{ for every }
 \mathfrak n\in\mathcal V_{\RA}(I)\right\}.
\]
Its normalized solid angle is
\[
 \omega_{2s}(\mathfrak m):=
 \frac{\sigma\bigl(N_{2s}(\mathfrak m)\cap S(\RA_{2s})\bigr)}
 {\sigma\bigl(S(\RA_{2s})\bigr)},
\]
where $S(\RA_{2s})$ is the unit sphere and $\sigma$ denotes its spherical
measure.

\begin{algorithm}[H]
\caption{Randomized computation of one real root}
\label{AlgorithmOneRoot}
\textbf{Input:} A finite $\A\subset\Z^n$, polynomials
$f_1,\dots,f_\mu\in\RA$ with finite real toric variety, and an integer
$s\ge1$.
\begin{itemize}
\item[(1)] Compute $\rho_\A$, set $D=\max_i\deg_\A(f_i)$, and put
$t=2\max\{D,\rho_\A,s\}$.
\item[(2)] Choose once $c$ uniformly on the unit sphere of $\RA_{2s}$.
\item[(3)] Minimize $\lf(c)$ over $\lf\in\K_t$. If $\K_t$ is empty, report
that there is no real toric solution. If the minimum is not attained, replace
$t$ by $t+2$ and repeat this step, keeping the same $c$.
\item[(4)] For an optimizer $\lf$, check whether
\[
 \rk\bigl(M_k(\lf)\bigr)=1
\]
for some $\max\{D,\rho_\A\}\le k\le t/2$. If so, return
\[
 \left\langle x^\alpha-\lf(x^\alpha):\alpha\in\A\right\rangle.
\]
Otherwise, replace $t$ by $t+2$ and return to step~(3).
\end{itemize}
\end{algorithm}

\begin{theorem}
\label{thm:one-root-correctness}
Algorithm~\ref{AlgorithmOneRoot} terminates with probability one and is
correct. Moreover, the probability of returning a solution $\mathfrak m$ is
precisely $\omega_{2s}(\mathfrak m)$ and is strictly positive.
\end{theorem}

\begin{proof}
If $\mathcal V_{\RA}(I)=\emptyset$, Lemma
\ref{TheoremCheckingExistenceOfRoots}(i) shows that $\K_t$ is empty for all
sufficiently large $t$. Conversely, the evaluation at a real point belongs to
every $\K_t$, so the algorithm cannot incorrectly report emptiness.

If step~(4) succeeds at any order, then
$\rk\bigl(M_k(\lf)\bigr)=\rk\bigl(M_{k-1}(\lf)\bigr)=1$.
Theorem~\ref{TheoremNtEqualsRealRadicalWithRankCondition}(i) and
Corollaries~\ref{CorollaryDirectFlatKernel} and
\ref{CorollaryDirectFlatAtomic} show that
$\langle\ker(M_k(\lf))\rangle$ is the maximal ideal of a real point of
$\mathcal V_{\RA}(I)$.
Thus every possible output of the algorithm is correct.

Assume that the variety is finite and nonempty, take $s_0,t_0$ as in
Theorem~\ref{thm:KtIsPolytope}, and put $q=\max\{s,s_0\}$. For every
sufficiently large $t$ occurring in the algorithm, $t\ge t_0+2q$, and hence
\[
 \pi_{2q}(\K_t)=\pi_{2q}(\K).
\]
Projecting this equality to $\RA_{2s}^*$ gives
$\pi_{2s}(\K_t)=\pi_{2s}(\K)$. Therefore the minimum in step~(3) is attained. The
objectives having more than one minimizer lie in a finite union of proper
linear subspaces, which has spherical volume zero. Hence, with probability
one, $c$ has a unique minimizer on $\pi_{2s}(\K)$, and Lemma
\ref{lem:quadratic-objectives-expose-roots} shows that this minimizer is
$\lf_{\mathfrak m}|_{\RA_{2s}}$ for some real point $\mathfrak m$.

Let $\lf\in\K_t$ be any optimizer. Its restriction to $\RA_{2q}$ belongs
to the simplex $\pi_{2q}(\K)$. Since its projection to $\RA_{2s}^*$ is the
vertex $\lf_{\mathfrak m}|_{\RA_{2s}}$ of $\pi_{2s}(\K)$, its simplex decomposition can
involve only $\lf_{\mathfrak m}$. Thus
\[
 \lf|_{\RA_{2q}}=\lf_{\mathfrak m}|_{\RA_{2q}},
 \qquad
 \rk\bigl(M_q(\lf)\bigr)=1.
\]
Step~(4) therefore succeeds.
Since $\lf$ agrees with $\lf_{\mathfrak m}$ on the generators of
$\RA$, Lemma~\ref{lem:Max_ideals_in_RA} shows that the returned ideal is
$\mathfrak m$.

Finally, outside the boundary of $N_{2s}(\mathfrak m)$, the solution
$\mathfrak m$ is returned exactly when $c\in N_{2s}(\mathfrak m)$. Since $c$
is uniform on the unit sphere, this event has probability
$\omega_{2s}(\mathfrak m)$. The strict separator $q_{\mathfrak m}$ constructed
in the lemma belongs to the interior of this cone, so this probability is
positive.
\end{proof}

\begin{remark}[Choice of $c$ and the norm]
If we choose $c\in\RA_1$, a solution may lie in the convex hull of the other
solutions after projection to $\RA_1^*$, and then its normal cone has solid
angle zero. Lemma~\ref{lem:quadratic-objectives-expose-roots} shows that degree
two is enough to give every solution strictly positive probability.
The actual probabilities depend on the norm chosen on $\RA_{2s}$.
\end{remark}

Increasing $s$ adds coordinates to the truncated evaluation vectors
and therefore changes the angles between the edges of the resulting
polytope. It does not follow that the probabilities become equal, or
that the probability of a fixed solution increases.
Moreover, affine translations of the points do not preserve the
solid angles neither.
The following
examples make both points explicit.

\begin{example}
\label{exp:unequal-limiting-solid-angles}
Let $\RA=\R[x]$, $I=\langle x(x-1)(x-2)\rangle$, and take the standard
monomial basis of $\RA_{2s}$ to be orthonormal. The evaluation vectors are
\[
 v_a^{(s)}=(1,a,\ldots,a^{2s}),\qquad a\in\{0,1,2\}.
\]
After projection to $\RA_1^*$, the evaluation at $1$ lies between the other
two evaluations, so its probability is zero for linear objectives. For
$s\geq1$, let $\theta_a^{(s)}$ be
the angle of the triangle formed by these vectors at $v_a^{(s)}$. Then
the probability at $a$ is $(\pi-\theta_a^{(s)})/(2\pi)$. A direct
computation of the scalar products between the edges shows that
$\theta_0^{(s)},\theta_1^{(s)}\to\pi/2$. Hence
$\theta_2^{(s)}\to0$, and therefore
\[
 (\omega_{2s}(0),\omega_{2s}(1),\omega_{2s}(2))
 \longrightarrow(1/4,1/4,1/2).
\]
In particular, increasing the objective degree does not make the
probabilities equal.
\end{example}

\begin{example}
\label{exp:translated-solid-angles}
Let $I=\langle x(x-1)(x+1)\rangle$, with the same orthonormal monomial basis.
For $a\in\{-1,0,1\}$, set $v_a^{(s)}=(1,a,\ldots,a^{2s})$. A direct
scalar-product computation gives the angles $\pi/4,\pi/2,\pi/4$, and hence
\[
 (\omega_{2s}(-1),\omega_{2s}(0),\omega_{2s}(1))
 =(3/8,1/4,3/8)
\]
for every $s\geq1$. Although $\{-1,0,1\}$ and $\{0,1,2\}$ have the same
pairwise distances, their probabilities are different. Thus these
probabilities depend on the geometry after the monomial embedding.
\end{example}

\begin{remark}[On numerical accuracy]
Even though theoretically we can recover the coordinates of the solutions associated to the vertex without needed to compute the border basis of the ideal, in some examples we observed that, because $\lf$ is an approximation, applying the eigenvalue method to reconstruct the corresponding extension lead to better numerical approximations. We observe that, in this case, the eigenvalue problem that we need to solve has (significantly) smaller size than the total number of real solutions.
\end{remark}

\begin{remark}
    [Relation to polynomial optimization]
    
We observe that this approach is closely related to  Lasserre’s hierarchy  \cite{lasserre2001global} for solving polynomial optimization problems over semialgebraic sets. Indeed, Lasserre’s hierarchy is based on the infinite-dimensional cone of moment functionals (in our case, $\K$), and considers its truncations to moments of degree at most $2t$ (which we called $\K_t$). This yields a sequence of semidefinite relaxations providing increasing lower bounds for the minimum of a polynomial function, since not every element of $\K_t$ admits an extension to an element of $\K$.
When the semialgebraic set is a set of points, Theorem~\ref{thm:KtIsPolytope} implies finite convergence for every fixed objective $c\in\RA_{2s}$: if $s\geq s_0$ and $t\geq t_0+2s$, then
\[
 \min_{\lf\in\K_t}\lf(c)=\min_{\lf\in\K}\lf(c).
\]
Indeed, the objective depends only on the restriction to
$\RA_{2s}$, and Theorem~\ref{thm:KtIsPolytope} gives
$\pi_{2s}(\K_t)=\pi_{2s}(\K)$. Notice that this does not assert that every
element of $\K_t$ extends to an element of $\K$; only its restriction to
$\RA_{2s}$ admits such an extension.
If $c$ has a unique minimizer on the real variety, the restriction
to $\RA_{2s}$ of every optimizer in $\K_t$ is the evaluation at this point.
\end{remark}

\paragraph{A unique nonnegative solution.}
Algorithm~\ref{AlgorithmOneRoot}  chooses an objective without restricting the signs of
its coefficients. The following observation is useful when one wants to find
a nonnegative solution by choosing an objective with nonnegative coefficients.

\begin{proposition}
\label{prop:unique-nonnegative-solution}
Let
$ \mathcal V_{\RA}(I)=\{\mathfrak m_+,\mathfrak m_1,\dots, \mathfrak m_{r-1}\} $
and define, for each solution, the vectors
\[
 v(\mathfrak m):=
 \bigl(\lf_{\mathfrak m}(x^{\alpha_1}),\dots,
 \lf_{\mathfrak m}(x^{\alpha_m})\bigr)\in\R^m.
\]
Assume that $v(\mathfrak m_+)\ge0$ coordinatewise and that every
$v(\mathfrak m_j)$ has a negative coordinate. Then there is
$c\in\RA_{r-1}$ with nonnegative coefficients such that
\[
 \lf_{\mathfrak m_+}(c)>
 \lf_{\mathfrak m_j}(c),
 \qquad j=1,\dots,r-1.
\]
Moreover, $c$ may be chosen with all its coefficients strictly positive.
\end{proposition}

\begin{proof}
  For every $j \in \{1,\dots,r-1\}$, choose $k(j) \in \{1,\dots,m\}$
  such that the $k(j)$-th coordinate of $v(\mathfrak m_j)$ is
  negative. Then, we define
  \[
    c :=
 \prod_{j=1}^{r-1}
 \left(x^{\alpha_{k(j)}} -v(\mathfrak m_j)_{k(j)}\right) \in\RA_{r-1}.
\]
The factor indexed by $j$ vanishes at $\mathfrak m_j$, and hence
$\lf_{\mathfrak m_j}(c)=0$. On the other hand, every factor is strictly
positive at $\mathfrak m_+$, so $\lf_{\mathfrak m_+}(c)>0$.

To make every coefficient strictly positive, fix a monomial basis $\mathcal B$
of $\RA_{r-1}$ and replace $c$ by
\[
 c_\varepsilon:=c+\varepsilon\sum_{b\in\mathcal B}b.
\]
For $\varepsilon>0$ sufficiently small, the strict inequalities above are
preserved. 
\end{proof}

\begin{remark}
  The previous proposition tells us that if an objective
  $c \in \RA_{r-1}$ is sampled from a distribution with positive
  density on the positive part of the unit sphere of $\RA_{r-1}$,
  minimizing the $-c$ returns $\mathfrak m_+$ with positive
  probability.
  However, the previous result is an existence one as we construct an
  specific polynomial $c$ using the coordinates of the solutions.
  It does not provide a procedure to compute such $c$.
\end{remark}

\subsection{Flat extension one root vs. all roots, an example}

We continue Example~\ref{exp:translated-solid-angles} to show that
flat extension at one point may occur earlier than flat extension at a generic
element.

\begin{example}[Continuation of Example~\ref{exp:translated-solid-angles}]
\label{exp:flat-extension-one-root-continuation}
At truncation degree $t=4$, we have that
 \[\mathcal{H}_4=\{x \, (x-1)(x-2) , x^2 \, (x-1)(x-2)\}, \qquad \text{and} \]
\begin{align*}
\K_4
= \Bigl\{\lf\in \R[x]_4^* \ \Big|\ & \lf(1)=1,\ M_2(\lf)\succeq 0,\ \lf(h_1)=0, \lf(x \cdot h_1 ) = 0\Bigr\} \\
= \Bigl\{\lf\in \R[x]_4^* \ \Big|\  &\lf(1)=1,\ \lf(x^2)\ge \lf(x)^2,\ \lf(x^3)=\lf(x), \lf(x^2) = \lf(x^4), \\
 & \underbrace{\lf(x^2)^2 + \lf(x^2)\cdot \lf(x)^2-\lf(x^2)^3-\lf(x)^2}_{= (\lf(x^2)- \lf(x))\cdot (\lf(x^2) + \lf(x)) \cdot (1 - \lf(x^2))} \geq 0 \Bigr\} 
 = \pi_4(\K).
\end{align*}
If we optimize the linear function $c=x^2$ we obtain
\begin{equation*}
    \mathrm{argmin}_{\lf \in \K_{{4}}} c(\lf) = \lf'
\end{equation*}
such that $\lf'(1)=1$ and
$\lf'(x)=\lf'(x^2)=\lf'(x^3)=\lf'(x^4)=0$, which gives
$\rk M_0(\lf')=\rk M_2(\lf')=1$. However, this cannot hold for generic
elements of $\K_4$. For example, consider
$\lf''(x)=\lf''(x^3)=0$ and
$\lf''(x^2)=\lf''(x^4)=2/3$. Then $\lf''$ lies in the relative interior of
$\K_4$, but $\rk M_0(\lf'')=1\ne3=\rk M_2(\lf'')$.
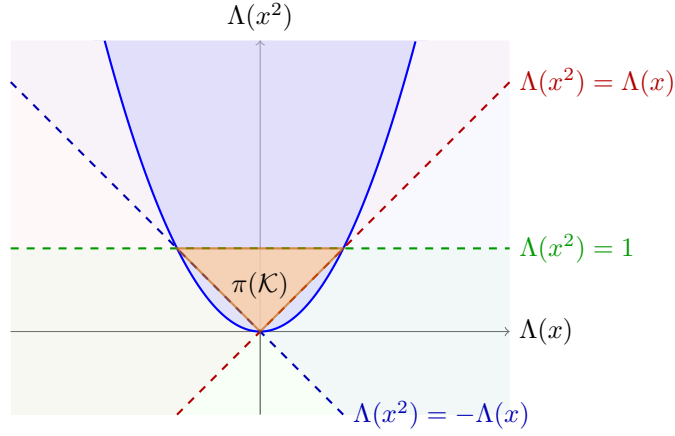
\begin{figure}
\begin{center}
\begin{tikzpicture}[scale=1.1]

    \def\xmin{-2}
    \def\xmax{2}
    \def\ymin{-1}
    \def\ymax{3.5}

    \draw[->] (-3,0) -- (3,0) node[right] {$\lf(x)$};
    \draw[->] (0,\ymin) -- (0,\ymax) node[above] {$\lf(x^2)$};

    \begin{scope}
        \clip (-3,\ymin) rectangle (3,\ymax);
        \fill[green!10,opacity=0.3] (-3,\ymin) rectangle (3,1);
    \end{scope}

    \begin{scope}
        \clip (-3,\ymin) rectangle (3,\ymax);
        \fill[red!10,opacity=0.25]
            (-3,-3) -- (3,3) -- (3,\ymax) -- (-3,\ymax) -- cycle;
    \end{scope}

    \begin{scope}
        \clip (-3,\ymin) rectangle (3,\ymax);
        \fill[blue!10,opacity=0.25]
            (-3,3) -- (3,-3) -- (3,\ymax) -- (-3,\ymax) -- cycle;
    \end{scope}

    \begin{scope}
        \clip (\xmin,\ymin) rectangle (\xmax,\ymax);
        \fill[blue!20,opacity=0.5]
            plot[domain=\xmin:\xmax,samples=200] (\x,{(\x)^2})
            -- (\xmax,\ymax)
            -- (\xmin,\ymax)
            -- cycle;
    \end{scope}

    \draw[thick,blue]
        plot[domain=-1.87:1.87,samples=200] (\x,{(\x)^2});

    \draw[thick,green!60!black,dashed]
        (-3,1) -- (3,1) node[right] {$\lf(x^2)=1$};

    \draw[thick,red!70!black,dashed]
        (-1,-1) -- (3,3) node[right] {$\lf(x^2)=\lf(x)$};

    \draw[thick,blue!70!black,dashed]
        (-3,3) -- (1,-1) node[right] {$\lf(x^2)=-\lf(x)$};

    \coordinate (A) at (-1,1);
    \coordinate (B) at (0,0);
    \coordinate (C) at (1,1);

    \filldraw[fill=orange!40,opacity=0.7,draw=orange!80!black,thick]
        (A) -- (B) -- (C) -- cycle;

    \node at (0,0.55) {$\pi(\mathcal{K})$};

\end{tikzpicture}
\end{center}
\caption{The set $\K_4$ is equal to $\pi(\K)$, which can be seen by the intersection of the halfspaces $ (\lf(x^2)- \lf(x)), (\lf(x^2) + \lf(x)) ,(1-\lf(x^2)) \geq 0$.}
\end{figure}
\end{example}

This example demonstrates something more interesting. Even if we have the polytope structure of $\K_t$ this does not imply that the generic elements have the flat extension property. In particular, every solution in the example above could be obtained with the right linear function, but not for the constant $1$ optimization, since we do not have flat extension for generic elements.

Next, we want to show that at all extreme points of $\K$ are also extreme points of $\K_t$, if $t \geq 2$. Hence, if we can always find the solutions of a polynomial system by shooting in the right direction. Of course, this is not practical, as shown by example \ref{exp:K_t_not_polyhedron}, since the normal cones of the extreme points of $\K_3$ corresponding to the solutions of $h_1$ have combined measure $0$. 

\begin{proposition}
\label{prop:Vertices_are_mapped_to_vertices}
Let $t\ge2$ be even. For every
$\mathfrak m\in\mathcal V_{\RA}(I)$, the truncated evaluation
$\lf_{\mathfrak m}|_{\RA_t}$ is an extreme point of $\K_t$.
\end{proposition}

\begin{proof}
Suppose
\[
 \lf_{\mathfrak m}|_{\RA_t}=\lambda\lf+(1-\lambda)\lf',
 \qquad 0<\lambda<1,
\]
with $\lf,\lf'\in\K_t$. Then we have that 
$
 M_{t/2}(\lf_{\mathfrak m})
 =\lambda M_{t/2}(\lf)+(1-\lambda)M_{t/2}(\lf').
$
Since $M_{t/2}(\lf_{\mathfrak m})$ has rank one, positivity gives
\[
 \ker\bigl(M_{t/2}(\lf_{\mathfrak m})\bigr)
 \subseteq\ker\bigl(M_{t/2}(\lf)\bigr)
 \cap\ker\bigl(M_{t/2}(\lf')\bigr).
\]
Thus each matrix on the right is a nonnegative multiple of
$M_{t/2}(\lf_{\mathfrak m})$. As all three functionals take the value $1$ at
$1$, these multiples are equal to $1$. Since $t$ is even, every monomial in $\RA_t$ is a
product of two monomials in $\RA_{t/2}$. Thus every value of the functionals
occurs as an entry of these matrices, and hence
$\lf=\lf'=\lf_{\mathfrak m}|_{\RA_t}$.
\end{proof}

The previous proposition shows that every real solution defines an
extreme point of $\K_t$ when $t$ is even. Thus, as in the preceding example,
it may be possible to recover a solution by choosing an appropriate objective
function, even if the generic elements of $\K_t$ do not satisfy the flat
extension property. However, the fact that the point is extreme does not
imply that it is obtained with positive probability when the objective is
chosen at random, since its normal cone may have solid angle zero.

%% file: Sections/Experiments.tex
\section{Experiments}
We perform numerical experiments to validate our approach in two cases.
We begin by analyzing sparse polynomial systems, comparing the dense versus sparse implementation of Algorithm \ref{AlgorithmMomentApproachForLx} in MATLAB \cite{Matlab}.
We show that exploiting sparsity leads to smaller moment matrices and faster termination.
Subsequently, we compare the accuracy of single-root extraction against all-root extraction using \texttt{MomentPolynomialOpt} \cite{MomentPolynomialOpt}, a Julia-based package \cite{bezanson2017julia}. All experiments were performed on a Lenovo IdeaPad 120S-14IAP with an Intel Pentium N4200 quad-core CPU (1.10 GHz), 4 GB RAM, and integrated Intel HD Graphics, running Windows 10 Home on an x86-64 architecture.

\subsection{Computational Experiments: Classical vs. Sparse}
\label{SectionComputationalResults}

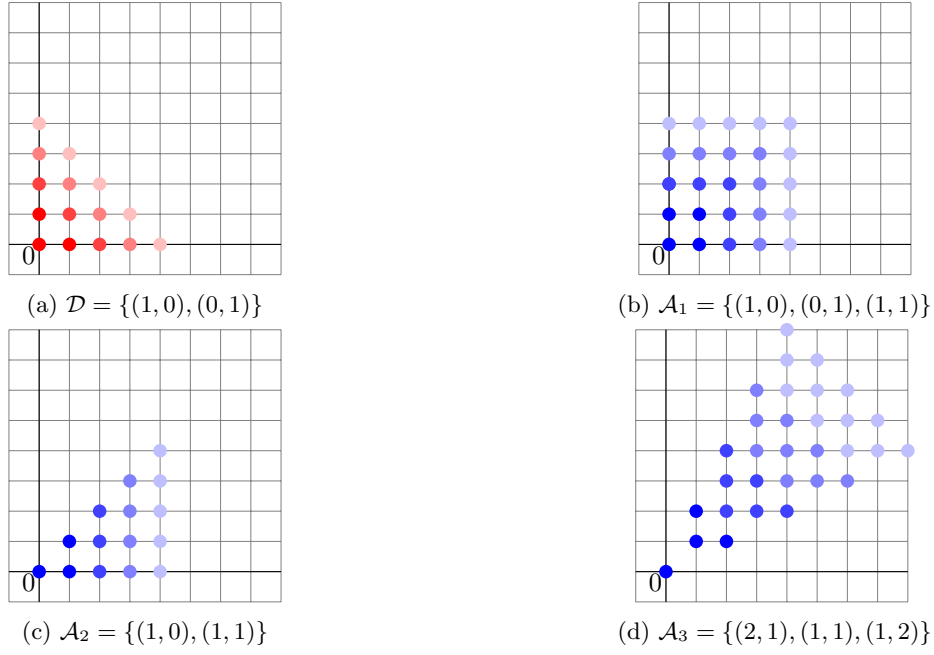
\begin{figure}[t]
\begin{subfigure}{0.5\textwidth}
\centering
\begin{tikzpicture}[scale = 0.4]
\draw[help lines] (-1,-1)grid(8,8);
\draw[] (-1,0)-- (8,0);
\draw[] (0,-1)-- (0,8);

\foreach \x/\y/\Name in {0/1/, 1/0/, 0/0/0}
\node[circle, fill= red!100,inner sep=1.8pt, 
label={[anchor=60, text=black]:$\Name$}
]at (\x,\y){};
\foreach \x/\y in {0/2/, 2/0/, 1/1}
\node[circle, fill= red!75,inner sep=1.8pt] at (\x,\y){};
\foreach \x/\y in {0/3/, 3/0/, 2/1,1/2}
\node[circle, fill= red!50,inner sep=1.8pt] at (\x,\y){};
\foreach \x/\y in {0/4/, 4/0/, 3/1, 1/3, 2/2}
\node[circle, fill= red!25,inner sep=1.8pt] at (\x,\y){};
\end{tikzpicture}
\caption{$\mathcal{D} = \{(1,0),(0,1)\}$}
  \label{fig:sub-first}
\end{subfigure}
\begin{subfigure}{.5\textwidth}
\centering
\begin{tikzpicture}[scale = 0.4]
\draw[help lines] (-1,-1)grid(8,8);
\draw[] (-1,0)-- (8,0);
\draw[] (0,-1)-- (0,8);

\foreach \x/\y/\Name in {1/0/, 1/1/, 0/1/, 0/0/0}
\node[circle, fill= blue!100,inner sep=1.8pt, 
label={[anchor=60, text=black]:$\Name$}
]at (\x,\y){};
\foreach \x/\y in {2/2/, 2/1/, 2/0,0/2, 1/2,0/2}
\node[circle, fill= blue!75,inner sep=1.8pt] at (\x,\y){};
\foreach \x/\y in {3/3/, 3/2/, 3/1,3/0,0/3,1/3,2/3}
\node[circle, fill= blue!50,inner sep=1.8pt] at (\x,\y){};
\foreach \x/\y in {4/4/, 4/3/, 4/2, 4/1, 4/0,0/4,1/4,2/4,3/4}
\node[circle, fill= blue!25,inner sep=1.8pt] at (\x,\y){};
\end{tikzpicture}
\caption{$\A_1 = \{(1,0),(0,1),(1,1)\}$}
\end{subfigure}

\begin{subfigure}{.5\textwidth}
\centering
\begin{tikzpicture}[scale = 0.4]
\draw[help lines] (-1,-1)grid(8,8);
\draw[] (-1,0)-- (8,0);
\draw[] (0,-1)-- (0,8);

\foreach \x/\y/\Name in {1/0/, 1/1/, 0/0/0}
\node[circle, fill= blue!100,inner sep=1.8pt, 
label={[anchor=60, text=black]:$\Name$}
]at (\x,\y){};
\foreach \x/\y in {2/2/, 2/1/, 2/0}
\node[circle, fill= blue!75,inner sep=1.8pt] at (\x,\y){};
\foreach \x/\y in {3/3/, 3/2/, 3/1,3/0}
\node[circle, fill= blue!50,inner sep=1.8pt] at (\x,\y){};
\foreach \x/\y in {4/4/, 4/3/, 4/2, 4/1, 4/0}
\node[circle, fill= blue!25,inner sep=1.8pt] at (\x,\y){};
\end{tikzpicture}
\caption{$\A_2 = \{(1,0),(1,1)\}$}
\end{subfigure}
\begin{subfigure}{.5\textwidth}
\centering
\begin{tikzpicture}[scale = 0.4]
\draw[help lines] (-1,-1)grid(8,8);
\draw[] (-1,0)-- (8,0);
\draw[] (0,-1)-- (0,8);

\foreach \x/\y/\Name in {2/1/, 1/1/, 1/2/, 0/0/0}
\node[circle, fill= blue!100,inner sep=1.8pt, 
label={[anchor=60, text=black]:$\Name$}
]at (\x,\y){};
\foreach \x/\y in {4/2/, 2/2/, 3/2,2/3,3/3, 2/4}
\node[circle, fill= blue!75,inner sep=1.8pt] at (\x,\y){};
\foreach \x/\y in {4/3/, 5/3,6/3,3/4,4/4,5/4,4/5,3/6,3/5}
\node[circle, fill= blue!50,inner sep=1.8pt] at (\x,\y){};
\foreach \x/\y in {4/8, 5/7,4/7,6/6,5/6,4/6,7/5,6/5,5/5,8/4,7/4,6/4}
\node[circle, fill= blue!25,inner sep=1.8pt] at (\x,\y){};
\end{tikzpicture}
\caption{$\A_3 = \{(2,1),(1,1),(1,2)\}$}
\end{subfigure}
\caption{The monomials up to degree four for each exponent set are drawn.}
\label{Fig:Exponent_Sets}
\end{figure}

In this section, we present computational experiments comparing Algorithm \ref{AlgorithmMomentApproachForLx} with the classical \cite[Algorithm 1]{laurent2012approach}. Our implementations, written in MATLAB \cite{Matlab}, are available on GitHub via \cite{GitHubSolver} and utilize the SDP solver SeDuMi \cite{SeDuMi}. Since this implementation is written as a prototype, we focus on small problems, specifically \( n=2 \) with two equations. The plots were generated using 20--30 instances out of a total of approximately 50--60 test cases.
The accuracy achieved by both algorithms was comparable, ranging
between \(10^{-3}\) and \(10^{-7}\). In practice, this tolerance could
not be improved further, as doing so often led to instability of the
SDP solver, particularly for dense instances. While higher accuracy
settings often succeeded for sparse instances, they rarely did so for
dense ones, resulting in significantly fewer successful runs and thus
reducing the number of meaningful comparisons. These numerical issues
stem from limitations that are inherent to SDP solvers, as discussed
in \cite{SeDuMiProblems}.
We emphasize that only instances for which at least one algorithm terminated successfully are included in the evaluation. Numerical difficulties were encountered frequently, with only about $50\%$ of all test cases terminating successfully. Moreover, whenever the dense approach terminated, the sparse approach almost always terminated as well, whereas the converse was not true. This indicates substantially poorer robustness of the dense approach.

For this reason, we postpone a detailed discussion of numerical errors to the next section, where we compare the strategies of computing all solutions versus computing only a single solution.

In the following, we consider three performance measures: the solving time of the largest SDP instance encountered by an algorithm, denoted by T(SDP); the total running time, denoted by T(total); and the dimension of the largest SDP instance. The first two quantities are compared using box plots, which also indicate outliers. To visualize the distribution of SDP dimensions, we use a clustered scatter plot in which the size of each circle is proportional to the number of experiments that produced the corresponding dimension.

We consider four different families of sparse systems with exponent sets with varying on degree and sparsity; see Figure \ref{Fig:Exponent_Sets}. The first one, 
\(\mathcal{D}\), corresponds to the dense case and serves as a baseline
for the exponent sets. 
The next one, $\A_1$, corresponds to systems with an affine
bihomogeneous structure.
The other two, \(\A_2\), and \(\A_3\), are families with monomials belonging to certain cones.
As discussed in Section~\ref{sec:realtoricvar}, the solution sets of the systems might change according to where we look for the solutions.
However, all these varieties will share the same torus $(\R^*)^n$. 
Under certain assumptions on the structure of the toric variety and for generic choices of the coefficients of the system, all the solutions will belong to this torus, so we can recover the same solutions regardless from where we look at them.
In particular, these assumptions hold for our families of random examples, so we do not go deeper into this issue.

We shall show with the figures that both the dimension and the number of iterations are significantly smaller for the sparse formulation than for the dense one. 
In particular, these values indicate how many and how difficult are the instances of SDP problems to solve; the smaller they are, the smaller the overall running times.
In these figures, we use the color blue to denote the behavior of the sparse algorithm and the red one for the dense.


We start our discussion by considering sparse polynomials supported on $\A_1$ for degrees two, three, four and five and compare the algorithms for the exponent vectors $\mathcal{D}$ (dense case, represented in red) and $\A_1$ (sparse case, represented in blue), see Figure \ref{fig:A_1}.
Overall, our sparsity-exploiting algorithm reduces problem size and computation time, often enhancing robustness.

\usepgfplotslibrary{groupplots}
\usepgfplotslibrary{statistics}


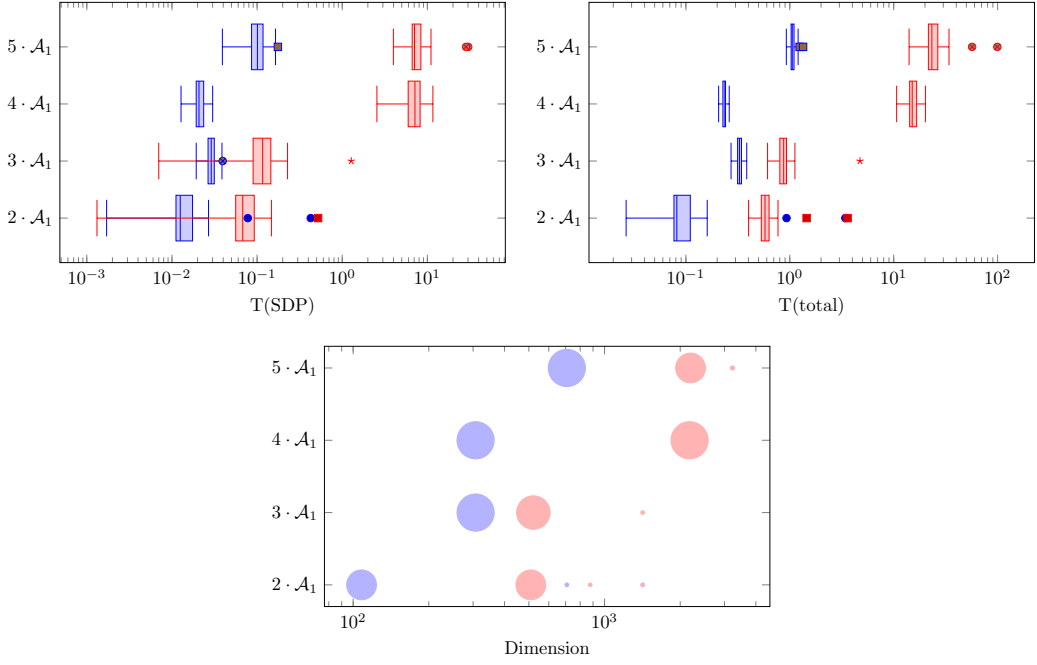
\begin{figure}[h]
\centering

\begin{tikzpicture}[scale = 0.7]

\pgfplotsset{
  boxplot A/.style={
    draw=blue,
    fill=blue!20,
  },
  boxplot B/.style={
    draw=red,
    fill=red!20,
  }
}

\begin{scope}[xshift=0cm]
\begin{axis}[
    width=10cm,
    height=6.5cm,
    boxplot/draw direction=x,
    xmode=log,
    log basis x=10,
    ytick={1,2,3,4},
    yticklabels={$2 \cdot \A_1$, $3 \cdot \A_1$, $4 \cdot \A_1$, $5 \cdot \A_1$},
    xlabel={T(SDP)},
    ylabel={},
    boxplot/every box/.style={solid, line width=0.4pt},
    boxplot/every whisker/.style={solid, line width=0.4pt},
    boxplot/every median/.style={solid, line width=0.4pt},
    boxplot/every average/.style={solid, line width=0.4pt},
    boxplot/every outlier/.style={
      only marks,
      mark=*,
      mark size=1.5pt,
    },
  ]

    \addplot+[
      boxplot prepared={
        lower quartile=0.011156,
        upper quartile=0.017454,
        median=0.012517,
        lower whisker=0.001709,
        upper whisker=0.026901
      },
      boxplot/draw position=1,
      boxplot A
    ] coordinates {
      (1,0.077755)
      (1,0.426787)
    };

    \addplot+[
      boxplot prepared={
        lower quartile=0.056085,
        upper quartile=0.092597,
        median=0.067855,
        lower whisker=0.00131700000000001,
        upper whisker=0.147365
      },
      boxplot/draw position=1,
      boxplot B
    ] coordinates {
      (1,0.520424)
    };

    \addplot+[
      boxplot prepared={
        lower quartile=0.026562,
        upper quartile=0.031413,
        median=0.028931,
        lower whisker=0.0192855,
        upper whisker=0.0386895
      },
      boxplot/draw position=2,
      boxplot A
    ] coordinates {
      (2,0.039469)
    };

    \addplot+[
      boxplot prepared={
        lower quartile=0.089846,
        upper quartile=0.145077,
        median=0.116379,
        lower whisker=0.00699949999999998,
        upper whisker=0.2279235
      },
      boxplot/draw position=2,
      boxplot B
    ] coordinates {
      (2,1.27706)
    };

    \addplot+[
      boxplot prepared={
        lower quartile=0.019273,
        upper quartile=0.023598,
        median=0.0207735,
        lower whisker=0.0127855,
        upper whisker=0.0300855
      },
      boxplot/draw position=3,
      boxplot A
    ] coordinates {};

    \addplot+[
      boxplot prepared={
        lower quartile=5.953789,
        upper quartile=8.21909025,
        median=7.150725,
        lower whisker=2.555837125,
        upper whisker=11.617042125
      },
      boxplot/draw position=3,
      boxplot B
    ] coordinates {};

    \addplot+[
      boxplot prepared={
        lower quartile=0.086132,
        upper quartile=0.117476,
        median=0.100869,
        lower whisker=0.039116,
        upper whisker=0.164492
      },
      boxplot/draw position=4,
      boxplot A
    ] coordinates {
      (4,0.17554)
    };

    \addplot+[
      boxplot prepared={
        lower quartile=6.632501,
        upper quartile=8.393343,
        median=7.069437,
        lower whisker=3.991238,
        upper whisker=11.034606
      },
      boxplot/draw position=4,
      boxplot B
    ] coordinates {
      (4,30.18719)
      (4,28.501574)
    };

  \end{axis}
\end{scope}

\begin{scope}[xshift=10cm]
\begin{axis}[
    width=10cm,
    height=6.5cm,
    boxplot/draw direction=x,
    xmode=log,
    log basis x=10,
    ytick={1,2,3,4},
    yticklabels={$2 \cdot \A_1$, $3 \cdot \A_1$, $4 \cdot \A_1$, $5 \cdot \A_1$},
    xlabel={T(total)},
    ylabel={},
    boxplot/every box/.style={solid, line width=0.4pt},
    boxplot/every whisker/.style={solid, line width=0.4pt},
    boxplot/every median/.style={solid, line width=0.4pt},
    boxplot/every average/.style={solid, line width=0.4pt},
    boxplot/every outlier/.style={
      only marks,
      mark=*,
      mark size=1.5pt,
    },
  ]

    \addplot+[
      boxplot prepared={
        lower quartile=0.076767,
        upper quartile=0.110222,
        median=0.081773,
        lower whisker=0.0265845,
        upper whisker=0.1604045
      },
      boxplot/draw position=1,
      boxplot A
    ] coordinates {
      (1,0.929868)
      (1,3.423639)
    };

    \addplot+[
      boxplot prepared={
        lower quartile=0.531767,
        upper quartile=0.631474,
        median=0.577068,
        lower whisker=0.4010915,
        upper whisker=0.7697035
      },
      boxplot/draw position=1,
      boxplot B
    ] coordinates {
      (1,3.600385)
      (1,1.451278)
    };

    \addplot+[
      boxplot prepared={
        lower quartile=0.314363,
        upper quartile=0.342645,
        median=0.327693,
        lower whisker=0.27194,
        upper whisker=0.385068
      },
      boxplot/draw position=2,
      boxplot A
    ] coordinates {};

    \addplot+[
      boxplot prepared={
        lower quartile=0.801589,
        upper quartile=0.928889,
        median=0.872876,
        lower whisker=0.610639,
        upper whisker=1.119839
      },
      boxplot/draw position=2,
      boxplot B
    ] coordinates {
      (2,4.73737)
    };

    \addplot+[
      boxplot prepared={
        lower quartile=0.22703075,
        upper quartile=0.24068625,
        median=0.2374635,
        lower whisker=0.2065475,
        upper whisker=0.2611695
      },
      boxplot/draw position=3,
      boxplot A
    ] coordinates {};

    \addplot+[
      boxplot prepared={
        lower quartile=14.2427185,
        upper quartile=16.6255875,
        median=15.0726845,
        lower whisker=10.668415,
        upper whisker=20.199891
      },
      boxplot/draw position=3,
      boxplot B
    ] coordinates {};

    \addplot+[
      boxplot prepared={
        lower quartile=1.029194,
        upper quartile=1.098228,
        median=1.055226,
        lower whisker=0.925643,
        upper whisker=1.201779
      },
      boxplot/draw position=4,
      boxplot A
    ] coordinates {
      (4,1.243324)
      (4,1.341967)
    };

    \addplot+[
      boxplot prepared={
        lower quartile=21.557136,
        upper quartile=26.558802,
        median=23.294806,
        lower whisker=14.054637,
        upper whisker=34.061301
      },
      boxplot/draw position=4,
      boxplot B
    ] coordinates {
      (4,99.360872)
      (4,56.825563)
    };

  \end{axis}
\end{scope}

\begin{scope}[shift={(5cm,-6.5cm)}]
  \newcommand{\bluepoint}[3]{%
    \addplot+[
      only marks,
      mark=*,
      mark size=#1,
      draw=blue!30,
      fill=blue!30,
    ] coordinates {(#2,#3)};
  }

  \newcommand{\redpoint}[3]{%
    \addplot+[
      only marks,
      mark=*,
      mark size=#1,
      draw=red!30,
      fill=red!30,
    ] coordinates {(#2,#3)};
  }

  \begin{axis}[
    width=10cm,
    height=6.5cm,
    xmode=log,
    log basis x=10,
    ytick={1,2,3,4},
    yticklabels={$2 \cdot \A_1$, $3 \cdot \A_1$, $4 \cdot \A_1$, $5 \cdot \A_1$},
    xlabel={Dimension},
    ylabel={},
    cycle list={}, 
  ]


    \bluepoint{8pt}{108}{1}   
    \bluepoint{1pt}{708}{1}
    \bluepoint{1pt}{1419}{1}
    \bluepoint{10pt}{307}{2}
    \bluepoint{10pt}{307}{3}
    \bluepoint{10pt}{708}{4}

    \redpoint{8pt}{509}{1}
    \redpoint{1pt}{1418}{1}
    \redpoint{1pt}{877}{1}
    \redpoint{9pt}{521}{2}
    \redpoint{1pt}{1418}{2}
    \redpoint{10pt}{2180}{3}
    \redpoint{8pt}{2198}{4}
    \redpoint{1pt}{3217}{4}
    \redpoint{1pt}{3237}{4}

  \end{axis}
\end{scope}

\end{tikzpicture}
\caption{Comparison of the sparse (blue) and dense (red) approaches,
  with respect to.~$\A_1$. The sparse approach consistently yields smaller SDP instances and significantly lower computation times, reducing T(SDP) by a factor of 4--10 and T(total) by a factor of 4--17.}
\label{fig:A_1}
\end{figure}



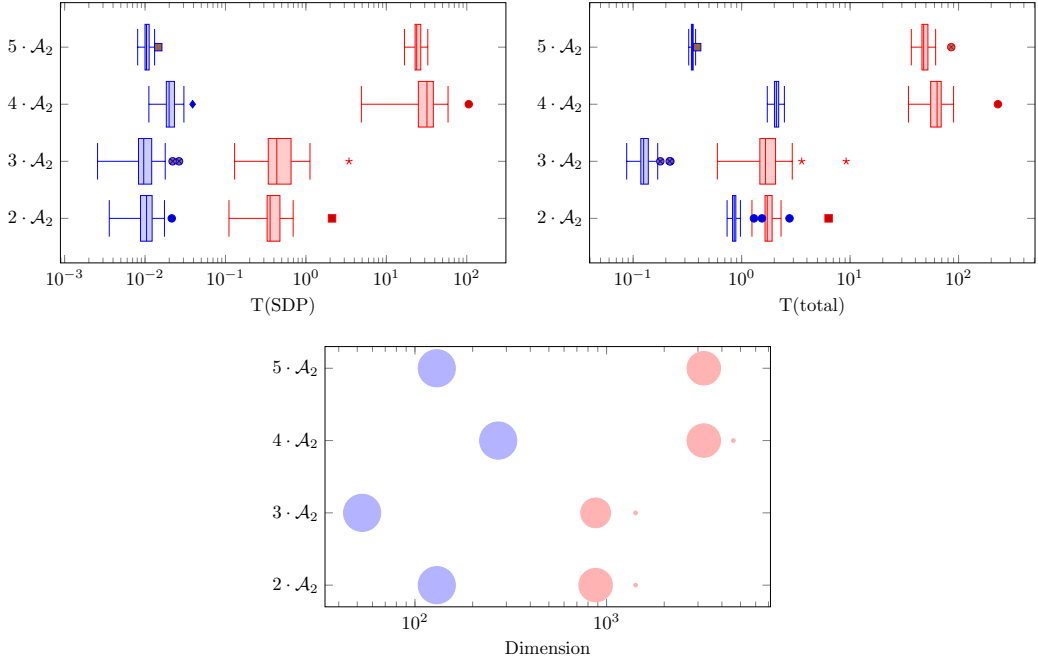
\begin{figure}[H]
\centering

\begin{tikzpicture}[scale = 0.7]

\pgfplotsset{
  boxplot A/.style={
    draw=blue,
    fill=blue!20,
  },
  boxplot B/.style={
    draw=red,
    fill=red!20,
  }
}

\begin{scope}[xshift=0cm]
\begin{axis}[
  width=10cm,
  height=6.5cm,
  boxplot/draw direction=x,
  xmode=log,
  log basis x=10,
  ytick={1,2,3,4},
  yticklabels={$2 \cdot \A_2$, $3 \cdot \A_2$, $4 \cdot \A_2$, $5 \cdot \A_2$},
  xlabel={T(SDP)},
  boxplot/every box/.style={solid, line width=0.4pt},
  boxplot/every whisker/.style={solid, line width=0.4pt},
  boxplot/every median/.style={solid, line width=0.4pt},
  boxplot/every average/.style={solid, line width=0.4pt},
  boxplot/every outlier/.style={
    only marks,
    mark=*,
    mark size=1.5pt,
  },
]


\addplot+[boxplot prepared={
  lower quartile=0.00878275,
  upper quartile=0.012243,
  median=0.0104445,
  lower whisker=0.003592375,
  upper whisker=0.017433375
}, boxplot/draw position=1, boxplot A]
coordinates {(1,0.021476)};

\addplot+[boxplot prepared={
  lower quartile=0.329455,
  upper quartile=0.4754145,
  median=0.358753,
  lower whisker=0.11051575,
  upper whisker=0.69435375
}, boxplot/draw position=1, boxplot B]
coordinates {(1,2.113677)};

\addplot+[boxplot prepared={
  lower quartile=0.00828625,
  upper quartile=0.01209875,
  median=0.00963,
  lower whisker=0.0025675,
  upper whisker=0.0178175
}, boxplot/draw position=2, boxplot A]
coordinates {(2,0.02646) (2,0.022011)};

\addplot+[boxplot prepared={
  lower quartile=0.3403435,
  upper quartile=0.65379125,
  median=0.4345445,
  lower whisker=0.1298281,
  upper whisker=1.123962875
}, boxplot/draw position=2, boxplot B]
coordinates {(2,3.434205)};

\addplot+[boxplot prepared={
  lower quartile=0.018381,
  upper quartile=0.0231955,
  median=0.0199115,
  lower whisker=0.01115925,
  upper whisker=0.03041725
}, boxplot/draw position=3, boxplot A]
coordinates {(3,0.039048)};

\addplot+[boxplot prepared={
  lower quartile=25.0023885,
  upper quartile=38.408275,
  median=32.0111529642857,
  lower whisker=4.89355875,
  upper whisker=58.51710475
}, boxplot/draw position=3, boxplot B]
coordinates {(3,105.908467)};

\addplot+[boxplot prepared={
  lower quartile=0.00998775,
  upper quartile=0.0112495166666667,
  median=0.010427,
  lower whisker=0.00809509999999999,
  upper whisker=0.0131421666666667
}, boxplot/draw position=4, boxplot A]
coordinates {(4,0.01459)};

\addplot+[boxplot prepared={
  lower quartile=22.7832585,
  upper quartile=26.77523875,
  median=23.8628935,
  lower whisker=16.795288125,
  upper whisker=32.763209125
}, boxplot/draw position=4, boxplot B]
coordinates {};


\end{axis}
\end{scope}

\begin{scope}[xshift=10cm]
\begin{axis}[
  width=10cm,
  height=6.5cm,
  boxplot/draw direction=x,
  xmode=log,
  log basis x=10,
  ytick={1,2,3,4},
  yticklabels={$2 \cdot \A_2$, $3 \cdot \A_2$, $4 \cdot \A_2$, $5 \cdot \A_2$},
  xlabel={T(total)},
  boxplot/every box/.style={solid, line width=0.4pt},
  boxplot/every whisker/.style={solid, line width=0.4pt},
  boxplot/every median/.style={solid, line width=0.4pt},
  boxplot/every average/.style={solid, line width=0.4pt},
  boxplot/every outlier/.style={
    only marks,
    mark=*,
    mark size=1.5pt,
  },
]


\addplot+[boxplot prepared={
  lower quartile=0.8245405,
  upper quartile=0.88457875,
  median=0.8546905,
  lower whisker=0.734483125,
  upper whisker=0.974636125
}, boxplot/draw position=1, boxplot A]
coordinates {(1,1.299376) (1,2.764311) (1,1.540671)};

\addplot+[boxplot prepared={
  lower quartile=1.6442425,
  upper quartile=1.9102395,
  median=1.722739,
  lower whisker=1.245247,
  upper whisker=2.309235
}, boxplot/draw position=1, boxplot B]
coordinates {(1,6.33949)};

\addplot+[boxplot prepared={
  lower quartile=0.1176465,
  upper quartile=0.13796125,
  median=0.1249495,
  lower whisker=0.087174375,
  upper whisker=0.168433375
}, boxplot/draw position=2, boxplot A]
coordinates {(2,0.220243) (2,0.17802) (2,0.217115)};

\addplot+[boxplot prepared={
  lower quartile=1.472416,
  upper quartile=2.0556135,
  median=1.658308,
  lower whisker=0.59761975,
  upper whisker=2.93040975
}, boxplot/draw position=2, boxplot B]
coordinates {(2,9.201767) (2,3.581409)};

\addplot+[boxplot prepared={
  lower quartile=2.007106,
  upper quartile=2.19573925,
  median=2.091512,
  lower whisker=1.724156125,
  upper whisker=2.478689125
}, boxplot/draw position=3, boxplot A]
coordinates {};

\addplot+[boxplot prepared={
  lower quartile=55.33304175,
  upper quartile=69.196781,
  median=63.4019397142857,
  lower whisker=34.537432875,
  upper whisker=89.992389875
}, boxplot/draw position=3, boxplot B]
coordinates {(3,230.571102)};

\addplot+[boxplot prepared={
  lower quartile=0.3445685,
  upper quartile=0.357396,
  median=0.3520065,
  lower whisker=0.32532725,
  upper whisker=0.37663725
}, boxplot/draw position=4, boxplot A]
coordinates {(4,0.388936)};

\addplot+[boxplot prepared={
  lower quartile=45.98361525,
  upper quartile=52.1849487,
  median=47.691679,
  lower whisker=36.68161505,
  upper whisker=61.4869489
}, boxplot/draw position=4, boxplot B]
coordinates {(4,85.650473)};


\end{axis}
\end{scope}

\begin{scope}[shift={(5cm,-6.5cm)}]
  \newcommand{\bluepoint}[3]{%
    \addplot+[
      only marks,
      mark=*,
      mark size=#1,
      draw=blue!30,
      fill=blue!30,
    ] coordinates {(#2,#3)};
  }

  \newcommand{\redpoint}[3]{%
    \addplot+[
      only marks,
      mark=*,
      mark size=#1,
      draw=red!30,
      fill=red!30,
    ] coordinates {(#2,#3)};
  }

  \begin{axis}[
    width=10cm,
    height=6.5cm,
    xmode=log,
    log basis x=10,
    ytick={1,2,3,4},
    yticklabels={$2 \cdot \A_2$, $3 \cdot \A_2$, $4 \cdot \A_2$, $5 \cdot \A_2$},
    xlabel={Dimension},
    ylabel={},
    cycle list={}, 
  ]


    \bluepoint{10pt}{130}{1}   
    \bluepoint{10pt}{53}{2}
    \bluepoint{10pt}{272}{3}
    \bluepoint{10pt}{130}{4}

    \redpoint{9pt}{877}{1}
    \redpoint{1pt}{1418}{1}
    \redpoint{8pt}{877}{2}
    \redpoint{1pt}{891}{2}
    \redpoint{1pt}{1418}{2}
    \redpoint{9pt}{3217}{3}
    \redpoint{1pt}{4589}{3}
    \redpoint{9pt}{3217}{4}
    \redpoint{1pt}{3237}{4}

  \end{axis}
\end{scope}

\end{tikzpicture}
\label{fig:A_2}
\caption{Comparison of the sparse (blue) and dense (red) approaches with respect to \(\A_2\). The sparse approach consistently outperforms the dense approach across all experiments, yielding smaller SDP instances and substantially reduced values of both T(SDP) and T(total).}
\end{figure}



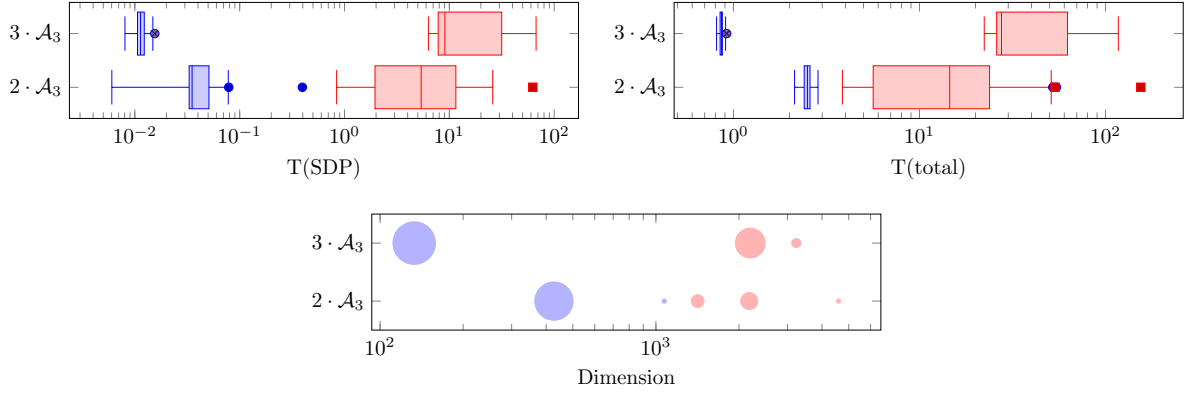
\begin{figure}[h]
\centering

\begin{tikzpicture}[scale = 0.8]

\pgfplotsset{
  boxplot A/.style={
    draw=blue,
    fill=blue!20,
  },
  boxplot B/.style={
    draw=red,
    fill=red!20,
  }
}

\begin{scope}[xshift=0cm]
\begin{axis}[
    width=10cm,
    height=3.5cm,
    boxplot/draw direction=x,
    xmode=log,
    log basis x=10,
    ytick={1,2},
    yticklabels={$2 \cdot \A_3$, $3 \cdot \A_3$},
    xlabel={T(SDP)},
    ylabel={},
    boxplot/every box/.style={solid, line width=0.4pt},
    boxplot/every whisker/.style={solid, line width=0.4pt},
    boxplot/every median/.style={solid, line width=0.4pt},
    boxplot/every average/.style={solid, line width=0.4pt},
    boxplot/every outlier/.style={
      only marks,
      mark=*,
      mark size=1.5pt,
    },
  ]

    \addplot+[
      boxplot prepared={
        lower quartile=0.032954,
        upper quartile=0.050902,
        median=0.035153,
        lower whisker=0.00603199999999999,
        upper whisker=0.077824
      },
      boxplot/draw position=1,
      boxplot A
    ] coordinates {
      (1,0.396299)
      (1,0.07866)
    };

    \addplot+[
      boxplot prepared={
        lower quartile=1.954369,
        upper quartile=11.537827,
        median=5.382997,
        lower whisker=0.84083,
        upper whisker=25.913014
      },
      boxplot/draw position=1,
      boxplot B
    ] coordinates {
      (1,62.512684)
    };

    \addplot+[
      boxplot prepared={
        lower quartile=0.010598,
        upper quartile=0.012299,
        median=0.01133,
        lower whisker=0.0080465,
        upper whisker=0.0148505
      },
      boxplot/draw position=2,
      boxplot A
    ] coordinates {
      (2,0.015488)
    };

    \addplot+[
      boxplot prepared={
        lower quartile=7.826219,
        upper quartile=31.572553,
        median=9.065106,
        lower whisker=6.309509,
        upper whisker=67.192054
      },
      boxplot/draw position=2,
      boxplot B
    ] coordinates {};

  \end{axis}
\end{scope}

\begin{scope}[xshift=10cm]
\begin{axis}[
    width=10cm,
    height=3.5cm,
    boxplot/draw direction=x,
    xmode=log,
    log basis x=10,
    ytick={1,2},
    yticklabels={$2 \cdot \A_3$, $3 \cdot \A_3$},
    xlabel={T(total)},
    ylabel={},
    boxplot/every box/.style={solid, line width=0.4pt},
    boxplot/every whisker/.style={solid, line width=0.4pt},
    boxplot/every median/.style={solid, line width=0.4pt},
    boxplot/every average/.style={solid, line width=0.4pt},
    boxplot/every outlier/.style={
      only marks,
      mark=*,
      mark size=1.5pt,
    },
  ]

    \addplot+[
      boxplot prepared={
        lower quartile=2.399732,
        upper quartile=2.578908,
        median=2.497972,
        lower whisker=2.130968,
        upper whisker=2.847672
      },
      boxplot/draw position=1,
      boxplot A
    ] coordinates {
      (1,52.081584)
      (1,54.423109)
    };

    \addplot+[
      boxplot prepared={
        lower quartile=5.635422,
        upper quartile=23.797081,
        median=14.529001,
        lower whisker=3.850995,    
        upper whisker=51.0395695
      },
      boxplot/draw position=1,
      boxplot B
    ] coordinates {
      (1,154.685061)
      (1,53.482702)
    };

    \addplot+[
      boxplot prepared={
        lower quartile=0.847338,
        upper quartile=0.871346,
        median=0.861764,
        lower whisker=0.811326,
        upper whisker=0.907358
      },
      boxplot/draw position=2,
      boxplot A
    ] coordinates {
      (2,0.919059)
    };

    \addplot+[
      boxplot prepared={
        lower quartile=25.985893,
        upper quartile=62.432756,
        median=27.584118,
        lower whisker=22.278503,   
        upper whisker=117.1030505
      },
      boxplot/draw position=2,
      boxplot B
    ] coordinates {};

  \end{axis}
\end{scope}

\begin{scope}[shift={(5cm,-3.5cm)}]
  \newcommand{\bluepoint}[3]{%
    \addplot+[
      only marks,
      mark=*,
      mark size=#1,
      draw=blue!30,
      fill=blue!30,
    ] coordinates {(#2,#3)};
  }

  \newcommand{\redpoint}[3]{%
    \addplot+[
      only marks,
      mark=*,
      mark size=#1,
      draw=red!30,
      fill=red!30,
    ] coordinates {(#2,#3)};
  }

 \begin{axis}[
    width=10cm,
    height=3.5cm,
    xmode=log,
    log basis x=10,
    ytick={1,2},
    yticklabels={$2 \cdot \A_3$, $3 \cdot \A_3$},
    ymin=0.5,      
    ymax=2.5,      
    xlabel={Dimension},
    ylabel={},
    cycle list={}, 
  ]


    \bluepoint{9pt}{427}{1}   
    \bluepoint{1pt}{1072}{1}
    \bluepoint{10pt}{133}{2}

    \redpoint{4pt}{2180}{1}
    \redpoint{3pt}{2198}{1}
    \redpoint{3pt}{1418}{1}
    \redpoint{1pt}{4589}{1}
    \redpoint{7pt}{2198}{2}
    \redpoint{2pt}{3217}{2}
    \redpoint{2pt}{3237}{2}

  \end{axis}
\end{scope}

\end{tikzpicture}

\caption{Comparison of the sparse (blue) and dense (red) approaches with respect to $A_3$. Once again, the sparse approach clearly outperforms the dense approach. In fact, the dense approach failed to solve any degree-four instance, whereas the sparse approach solved all such instances without difficulty. Consequently, no direct comparison for degree-four instances could be included in the figure.
}
\end{figure}

\vspace{\baselineskip}

\subsection{Computing one root vs. all roots}
To evaluate the accuracy of computing a single root versus all roots, we conducted experiments in the dense setting using the Julia package \cite{MomentPolynomialOpt}. For the underlying optimization, we employed the COSMO solver \cite{Garstka_2021}, selected for its ADMM-based approach, which is particularly effective for handling large SDP instance sizes with lower memory requirements compared to interior-point methods. This makes COSMO especially suitable for the size of the SDP problems in our experiments.


The results are suboptimal, as we do not perform a flat extension check; instead, we heuristically select the truncation degree \( t \) (see Algorithm \ref{AlgorithmMomentApproachForLx}) and proceed. Moreover, multiple reruns were often required because the solver frequently found more than one solution---an unexpected outcome, since the optimization over a spectrahedron is typically unique.

In order to measure the numerically accuracy, we use the backward error (BWE), see \cite[Appendix C]{TelenPhD}, of a numerical approximation $\lf_\mathfrak m$ of a solution contained in $\mathcal{V}_{\RA}(f_1,\dots,f_m)$. This measure is defined as 
\begin{equation}
\label{eq:BWE}
    \mathrm{BWE}(\lf_\mathfrak m) = \frac{1}{m} \sum_{i = 1}^m \frac{|\lf_\mathfrak m (f_i)|}{\sum_{\alpha \in \A_\infty}|c_{\alpha,f} \lf_\mathfrak m(x^\alpha)| + 1},
\end{equation}
where $f_i = \sum_{\alpha \in \A_\infty} c_{\alpha,f} x^\alpha$.


We present tables comparing the following metrics: 
\begin{center}
\begin{tabular}{ |p{2.5cm}|p{12cm}| }
\hline
\multicolumn{2}{|c|}{Description of metrics} \\
\hline
D & Degree of the dense polynomials. \\
\hline
\#vars & Number of variables. \\
\hline
t & Truncation degree used in the computations. \\
\hline
time & Total solver time in seconds. \\
\hline
BWE & The BWE computed by applying a random linear functional, compared to the BWE using a constant 1 functional. \\
\hline
\#sol & Number of solutions found in both approaches. \\
\hline
\end{tabular}
\end{center}



The input for the software \texttt{MomentPolynomialOpt.jl} has the following structure. 
Suppose we want to solve the polynomial system \( f_1,\dots,f_m \in \mathbb{R}[x_1,\dots,x_n] \). 
We choose a truncation degree \( t \ge \max_{k=1,\dots,m} \deg(f_k) \) and generate a random polynomial \( c \in \Rx_t \). 
We then consider the following inputs:
\begin{align*}
\texttt{const:}\quad & 
    v_{\texttt{const}},\, M_{\texttt{const}}
    = \texttt{optimize}(:\mathrm{min},\, 1\,,\,[f_1,\dots,f_m],\, [],\,[x_1,\dots,x_n],\, t,\, \mathrm{COSMO}), \\
\texttt{rnd:}\quad &
    v_{\texttt{rnd}} \;\;\;,\, M_{\texttt{rnd}} \;\,\,
    = \texttt{optimize}(:\mathrm{min},\, c\;,\,[f_1,\dots,f_m],\, [],\,[x_1,\dots,x_n],\, t,\, \mathrm{COSMO}).
\end{align*}
The function \texttt{optimize} returns the optimal values \( v_{\texttt{const}} \) and \( v_{\texttt{rnd}} \) for the constant objective \(1\) and for the random objective \(c\), respectively. 
Using the package's function \texttt{get\_measure}, we can recover the solution set of
\[
    f_1 = \dots = f_m = 0
\]
with respect to the chosen objective. 
This can be achieved with the following prompt:
\begin{align*}
    &\_,\, X_{\texttt{const}} = \texttt{get\_measure}(M_{\texttt{const}}), \\
    &\_,\, X_{\texttt{rnd}} \;\;\,  = \texttt{get\_measure}(M_{\texttt{rnd}}).
\end{align*}
The matrices \( X_{\texttt{const}} \) and \( X_{\texttt{rnd}} \) contain, as column vectors, the minimizers of the objective functions \(1\) and \(c\), subject to the constraints \( f_1 = \dots = f_m = 0 \). 
If \( X_{\texttt{rnd}} \) contains exactly one column, we select from \( X_{\texttt{const}} \) the column that minimizes the Euclidean distance to the column of \( X_{\texttt{rnd}} \), and we compare the corresponding BWEs.

\begin{figure}[H]
\centering

\begin{subtable}[t]{0.44\textwidth}
  \centering
  \caption{First test}
  \resizebox{\textwidth}{!}{\begin{tabular}{lccccc}
  \toprule
  (D,\#vars,t) & \multicolumn{2}{c}{T(total)} & \multicolumn{1}{c}{$\approx \mathrm{log}_{10} (\mathrm{BWE})$} & \multicolumn{2}{c}{\#sol}  \\
  \cmidrule(lr){2-3} \cmidrule(lr){4-4} \cmidrule(lr){5-6} 
  & rnd & const  & diff & rnd & const  \\
  \midrule
  $(5,2,9)$ & 0.319 & 0.283 & 0.852 & 1 & 2.558  \\
  $(5,2,12)$ & 1.216 & 0.783 & 0.558 & 1 & 2.558 \\
  $(5,2,15)$ & 4.530 & 3.185 & 0.820 & 1 & 2.538  \\
  $(5,2,17)$ & 7.851 & 8.428 & 0.675 & 1 & 2.3 \\
  \bottomrule
  \end{tabular}}
  \label{tab:test1}
\end{subtable}
\hfill
\begin{subtable}[t]{0.44\textwidth}
  \centering
  \caption{Second test}
  \resizebox{\textwidth}{!}{\begin{tabular}{lccccc}
  \toprule
(D,\#vars,t) & \multicolumn{2}{c}{T(total)} & \multicolumn{1}{c}{$\approx \mathrm{log}_{10} (\mathrm{BWE})$} & \multicolumn{2}{c}{\#sol}  \\
\cmidrule(lr){2-3} \cmidrule(lr){4-4} \cmidrule(lr){5-6} 
& rnd & const  & diff & rnd & const  \\
\midrule
$(7,2,12)$ & 2.182 & 1.495 & 2.081 & 1 & 5.513  \\
$(7,2,15)$ & 4.094 & 3.646 & 1.241 & 1 & 3.379 \\
$(7,2,17)$ & 10.533 & 8.672 & 1.051 & 1 & 3.692  \\
\bottomrule
  \end{tabular}}
  \label{tab:test2}
\end{subtable}

\vspace{1em}

\begin{subtable}[t]{0.44\textwidth}
  \centering
  \caption{Third test}
  \resizebox{\textwidth}{!}{\begin{tabular}{lccccc}
  \toprule
(D,\#vars,t) & \multicolumn{2}{c}{T(total)} & \multicolumn{1}{c}{$\approx \mathrm{log}_{10} (\mathrm{BWE})$} & \multicolumn{2}{c}{\#sol}  \\
\cmidrule(lr){2-3} \cmidrule(lr){4-4} \cmidrule(lr){5-6} 
& rnd & const  & diff & rnd & const  \\
\midrule
$(8,2,12)$ & 2.016 & 0.977 & 1.571 & 1 & 5.607  \\
$(8,2,14)$ & 4.241 & 3.422 & 1.617 & 1 & 6.176 \\
$(8,2,15)$ & 5.451 & 4.001 & 1.833 & 1 & 5.75  \\
\bottomrule
  \end{tabular}}
  \label{tab:test3}
\end{subtable}
\hfill
\begin{subtable}[t]{0.44\textwidth}
  \centering
  \caption{Fourth test}
  \resizebox{\textwidth}{!}{\begin{tabular}{lccccc}
  \toprule
(D,\#vars,t) & \multicolumn{2}{c}{T(total)} & \multicolumn{1}{c}{$\approx \mathrm{log}_{10} (\mathrm{BWE})$} & \multicolumn{2}{c}{\#sol}  \\
\cmidrule(lr){2-3} \cmidrule(lr){4-4} \cmidrule(lr){5-6} 
& rnd & const  & diff & rnd & const  \\
\midrule
$(9,2,13)$ & 3.489 & 1.839 & 2.842 & 1 & 9.631  \\
$(9,2,14)$ & 4.644 & 3.366 & 1.482 & 1 & 8.103 \\
$(9,2,15)$ & 6.614 & 4.181 & 1.928 & 1 & 7.071  \\
$(9,2,16)$ & 9.271 & 5.661 & 1.758 & 1 & 8.586  \\
$(9,2,17)$ & 11.163 & 9.511 & 1.576 & 1 & 7.538  \\
\bottomrule
  \end{tabular}}
  \label{tab:test4}
\end{subtable}

\vspace{1em}

\begin{subtable}[t]{0.44\textwidth}
  \centering
  \caption{Fifth test}
  \resizebox{\textwidth}{!}{\begin{tabular}{lccccc}
  \toprule
(D,\#vars,t) & \multicolumn{2}{c}{T(total)} & \multicolumn{1}{c}{$\approx \mathrm{log}_{10} (\mathrm{BWE})$} & \multicolumn{2}{c}{\#sol}  \\
\cmidrule(lr){2-3} \cmidrule(lr){4-4} \cmidrule(lr){5-6} 
& rnd & const  & diff & rnd & const  \\
\midrule
$(10,2,13)$ & 2.287 & 2.221 & 2.541 & 1 & 14.583  \\
$(10,2,14)$ & 5.650 & 3.161 & 2.235 & 1 & 13.235 \\
$(10,2,15)$ & 6.089 & 4.054 & 2.4 & 1 & 14.2  \\
$(10,2,16)$ & 9.055 & 6.751 & 2.217 & 1 & 9.739  \\
\bottomrule
  \end{tabular}}
  \label{tab:test5}
\end{subtable}
\caption{For each instance of (D, {\#}vars, t), we ran the algorithms 60 times on independently generated random polynomial systems to collect the data. The comparison is restricted to cases where the random approach computed exactly one solution, which happened in approximately 50--60{\%} of all runs. The column \emph{diff} reports the improvement in accuracy of the BWE for the random approach relative to the constant approach. Since we restrict our attention to the setting in which the random approach yields exactly one solution, the difference corresponds precisely to this unique solution.
}
\label{fig:const_vs_rnd}
\end{figure}

We observe in Figure \ref{fig:const_vs_rnd} that the accuracy of a given solution can be improved. However, the improvement is relatively modest, which we suspect is due to limitations of the SDP solvers. Another difficulty is that the random approach frequently computes more than one solution, which limits the number of directly comparable instances.

%% file: Sections/Appendix.tex
\appendix

\section{Solving Sparse Polynomial Systems}

In this appendix, we adapt border-basis to the semigroup framework
needed in the paper and explain the corresponding eigenvalue
method.

\subsection{Sparse Border Bases}
\label{sec:border_basis}

We want to generalize border bases to semigroups generated by a finite
subset $\A = \{\alpha_1,\dots,\alpha_m\}$ contained in
$\mathbb{Z}^n$. This is a generalization of the theory developed in
\cite{ToricBorderBases}.  More precisely, Mourrain and Tr\'ebuchet
introduced toric border bases in a Laurent polynomial ring, using
multiplication by the variables and their inverses
\cite{ToricBorderBases}. Here we adapt their definitions to the
semigroup algebra $\RA$, where the generators are the monomials
$x^{\alpha_i}$ and the multiplication maps must satisfy all the
relations in $\ker(\psi_\A)$. The construction also builds on the
normal-form criterion of Mourrain \cite{NormalForms}.  We begin with
the following definition.  The prolongation and border below
  extend the corresponding notions from
  \cite[Section~1]{ToricBorderBases}, replacing multiplication by
  $x_i^{\pm1}$ with multiplication by the generators $x^{\alpha_i}$.

\begin{definition}
Let $\A = \{\alpha_1,\dots,\alpha_m\}$ be a finite subset contained in $\mathbb{Z}^n$.
Given a set $\B \subset \{x^{\alpha} : \alpha \in \A_\infty\}$ and a linear subspace $\Lin \subset \RA$, define the new sets 
\begin{equation*}
\B_{\A}^{+} \coloneqq \B \cup \bigcup_{i = 1}^m x^{\alpha_i} \B = \B \cup \{x^{\alpha_i} b : b \in \B, i=1,\dots,m\}, \quad \partial \B_{\A} = \B_{\A}^+ \setminus \B,
\end{equation*}
called, respectively, the prolongation of $\B$ and the border of $\B$ with respect to $\A$, as well as the linear subspace
\begin{equation*}
\Lin_{\A}^+ \coloneqq \Lin + x^{\alpha_1} \Lin + \dots + x^{\alpha_m} \Lin.
\end{equation*}
Note that for a monomial subset $\B$ it holds $\mathrm{lin}(\B)_{\A}^+ = \lin(\B_{\A}^+)$.
We denote by $\B_{\A}^d \coloneqq (\B_{\A}^{d-1})_{\A}^+$, the $d$-th power of the operator $'+'$ on $\B$. 
By convention, $\B_{\A}^0 = \B$ and we also write $\B_{\A}^* \coloneqq \bigcup_{d \geq 0} \B_{\A}^d$.
Analogously define $\Lin_{\A}^d \coloneqq (\Lin_{A}^{d-1})_{\A}^+$, where $\Lin_{\A}^0 = \Lin$ and $\Lin_{\A}^* \coloneqq \cup_{d \geq 0} \Lin_{\A}^d$.
It is easy to verify that $\Lin_{\A}^* = \langle \Lin \rangle$, the ideal generated by $\Lin$.
\label{DefinitionProlongationAndBorderWithA}
\end{definition}

We further need to generalize the connectivity property.
Therefore, we need the degree with respect to $\A$, as defined in Definition \ref{DefinitionRA}, which is essentially the same as the index of a polynomial with respect to the subspace $\Lin \coloneqq \lin(1)$ and $\A$.
The monomial condition is the corresponding adaptation of the set
connected to $1$ in \cite[Section~1]{ToricBorderBases}.

\begin{definition} \label{def:connected-to-one}
Let $\B$ be a monomial set.
The set $\B$ is said to be connected-to-one over $\mathcal{A}$ if $1 \in \B$ and each $m \in \B \setminus \{1\}$ can be written as $m = x^{\alpha_{i1}} \cdots x^{\alpha_{ik}}$ with $x^{\alpha_{i1}},x^{\alpha_{i1}} x^{\alpha_{i2}},\dots,x^{\alpha_{i1}} \cdots x^{\alpha_{ik}} \in \B$, and $\alpha_{il} \in \A$.
We say a linear subset $\Lin \subseteq \RA$ is connected to $1 \in \Lin$, if for any $f \in \Lin$ with $\deg_{\A}(f) > 0$, there exists $f_1,\dots,f_m \in \Lin$ such that
\begin{equation}
f = \sum_{i=1}^m x^{\alpha_i} f_i,
\label{EquationConnectedTo1}
\end{equation}
with $f_i$ a multiple of $1$ and $\deg_{\A}(f_i) < \deg_{\A}(f)$.
\label{DefinitionConnectedTo1}
\end{definition}

To generalize the multiplication operators, as we deal with finite subsets $\A \subset \mathbb{Z}^n$, we need some additional constraints for the projection $N$, taking the geometry of $\A$ into account.
This is the analogue of the commutation and inversion conditions in
\cite[Theorem~2.3]{ToricBorderBases}. For a general semigroup algebra, these
conditions are replaced by compatibility with every relation among the
generators.

\begin{definition}
\label{DefinitionMiNProjectionGeneralized}
Consider a finite subset $\A = \{\alpha_1,\dots,\alpha_m \} \subset \mathbb{Z}^n$.
Let $N: \mathcal{L}^+ \rightarrow \mathcal{L}$ be a linear map.
We say that $N$ is a projection with respect to $\A$, if the following two conditions are satisfied:
\begin{itemize}
\item[i)] $N \cdot N = N$,
\item[ii)] Whenever elements $\beta_1,\dots,\beta_j,\gamma_1,\dots,\gamma_k \in \A$ satisfy $\sum_{i = 1}^j \beta_i = \sum_{l = 1}^k \gamma_l$, the following identity holds for all $f \in \Lin$
\begin{equation*}
N(x^{\beta_1} N( \dots N(x^{\beta_{j-1}} N( x^{\beta_j} f)) \dots ) = N(x^{\gamma_1} N( \dots N(x^{\gamma_{k-1}} N( x^{\gamma_k} f)) \dots ).
\end{equation*}
\end{itemize}
Further, for a projection $N$ with respect to $\A$, and $i \in \{1,\dots,m\}$ we define the \emph{multiplication operator} $\mathcal{M}_i$, as the linear map
\begin{equation*}
\mathcal{M}_i : \mathcal{L} \rightarrow \mathcal{L}, \quad f \mapsto N(x^{\alpha_i} f).
\end{equation*}
\end{definition}

The second condition in Definition \ref{DefinitionMiNProjectionGeneralized} for a projection $N$ with respect to $\A$ ensures that we can define the object $f(\mathcal{M})$, $f \in \RA$, despite the fact that we have $m = |\A|$ multiplication maps instead of $k = \dim(\A_\infty)$.
Recall the map 
\begin{equation*}
\psi : \R[x_1,\dots,x_m] \rightarrow \RA, \quad x_i \mapsto x^{\alpha_i},
\end{equation*}
which is surjective.
Now for a polynomial $p = \sum_{\beta \in \mathbb{N}^m} p_\beta x^\beta \in \ker(\psi)$ it holds that $p(\mathcal{M}) = 0$, where $p(\mathcal{M}) = \sum_{\beta \in \mathbb{N}^m} p_\beta \mathcal{M}_1^{\beta_1} \cdots \mathcal{M}_m^{\beta_m}$.
Thus, the object $f(\mathcal{M}) \coloneqq \psi^{-1}(f)(\mathcal{M})$, where $f \in \RA$, is well-defined and does not depend on the chosen element of $\psi^{-1}(f)$.

With this notation, the normal-form and border-basis criteria of
\cite{NormalForms,ToricBorderBases} carry over once one verifies that
the induced multiplication maps factor through the toric ideal
$\ker(\psi_\A)$. We refer to these papers for the original statements
and proofs.


\subsection{Sparse Eigenvalue Method}
\label{SectionEigenvalue}

The theory about sparse border bases allows us to construct a basis for the quotient space $\RA / \rIRA$ and thus the computation of the multiplication operators $\mathcal{M}_i$.
The next theorem is a generalization of \cite[Proposition  4.7]{UsingAlgGeo} and will help us in order to solve a polynomial system by only computing the eigenvalues and eigenvectors of a single multiplication map $m_f: \RA/\rIRA \rightarrow \RA/\rIRA, h \mapsto h \cdot f$, for which the polynomial $f \in \RA$ needs to fulfill some extra properties.
This is the eigenvalue method for polynomial roots, also called the
Stickelberger eigenvalue theorem; see, in particular,
\cite{Mourrain1998MatrixMethods} and \cite[Theorem~3]{laurent2012approach}. We give
the short proof in the notation of the semigroup algebra.

\begin{theorem}
Let $I \subseteq \RA$ be a real radical ideal with $1 \leq r = |\mathcal{V}_{\RA}(I)| < \infty$.
Suppose $f \in \RA$ is chosen such that the values $\lf_{\mathfrak m}(f)$ are distinct for $\mathfrak m \in \VL(I)$.
Further, let $\B = \{x^{\beta_1},\dots,x^{\beta_r} \}$ be a basis of $\RA / I$, where $\beta_i \in \A_\infty$.
Then the left eigenspaces of the matrix $m_f$ are $1$-dimensional and are spanned by the row vectors $(\lf_{\mathfrak m}(x^{\beta_1}),\dots,\lf_{\mathfrak m}(x^{\beta_r}))$ for $\mathfrak m \in \VL(I)$.
\label{TheoremStickelberger2}
\end{theorem}

\begin{proof}
By Lemma~\ref{lem:Real_Nullstellensatz_Toric_variety}, evaluation induces the
isomorphism
\[
 \Theta:\RA/I\longrightarrow\R^r,
 \qquad [g]\longmapsto
 \bigl(\lf_{\mathfrak m}(g)\bigr)_{\mathfrak m\in\mathcal V_{\RA}(I)}.
\]
Under this isomorphism, multiplication by $f$ corresponds to the diagonal
operator with entries $\lf_{\mathfrak m}(f)$. Hence the rows of the matrix of
$\Theta$ in the basis $\B$, namely
\[
 \bigl(\lf_{\mathfrak m}(x^{\beta_1}),\dots,
 \lf_{\mathfrak m}(x^{\beta_r})\bigr),
\]
are left eigenvectors of $m_f$. Since $\Theta$ is an isomorphism and the
values $\lf_{\mathfrak m}(f)$ are pairwise distinct, these are all the left
eigenspaces and each has dimension one.
\end{proof}

Theorem \ref{TheoremStickelberger2}, which is the generalized version of the Stickelberger Theorem \cite[Theorem 3]{laurent2012approach}, allows us to apply the eigenvalue method for real root-finding problems.
We want to explain how to use Theorem \ref{TheoremStickelberger2} in order to find all  the points $\VL(I)$, if $\VL(I)$ is finite.
Recovering the roots from the coordinates of the eigenvectors is the
standard matrix method described by Mourrain
\cite{Mourrain1998MatrixMethods}. The adaptation below recovers instead the
values of the semigroup generators $x^{\alpha_i}$.

Write $\A = \{\alpha_1,\dots,\alpha_m \}$ and assume that the ideal $I$ is real radical by replacing $I$ with $\rIRA$.
Then, we are interested in a basis $\B$ of the quotient space.
Here $\B = \{ x^{\beta_1},\dots,x^{\beta_r} \}$ is connected-to-one and can be computed by the theory of border basis.
This basis $\B \subset \RA$ is connected-to-one.
Notice that we assume $1 \in \B$ which can only happen if and only if $\VL(I) \neq \emptyset$, so we rule out the trivial case.
Now consider a random linear function 
\begin{equation*}
f = c_1 x^{\alpha_1} + \cdots + c_m x^{\alpha_m},
\end{equation*}
where $c_1,\dots,c_m$ are randomly chosen integers.
Almost every choice of $c_1,\dots,c_m$ will ensure that the values $\lf_{\mathfrak m}(f)$ for $\mathfrak m \in \VL(I)$ are distinct. 
We are able to construct the matrix $m_f$ relative to the monomial basis $\B$ by linear algebra.
This in turn allows us to compute the eigenvalues $\lambda$ and eigenvectors $w$ of the multiplication operator $m_f$.
An eigenvector $w$ combined with the basis $\B$ makes it possible to find the generators $x^{\alpha_i} - \lf_{\mathfrak m}(x^{\alpha_i})$ for a solution $\mathfrak m \in \VL(I)$, cf. Lemma \ref{lem:Max_ideals_in_RA}.

To see how this is done, note that Theorem \ref{TheoremStickelberger2} implies 
\begin{equation}
\label{EquationEigenvectorSolution}
w = c\cdot (\lf_{\mathfrak m}(x^{\beta_1}),\dots,\lf_{\mathfrak m}(x^{\beta_r})),
\end{equation}
for some nonzero constant $c$ and some $\mathfrak m \in \VL(I)$.
Our goal is to compute the generators of $\mathfrak m$ in terms of the coordinates of $w$.
Equation (\ref{EquationEigenvectorSolution}) implies that each coordinate of $w$ is of the form $c \cdot \lf_{\mathfrak m}(x^{\beta_j})$.

The computation of $\mathfrak m$ goes as follows. 
First of all, we have $[1] \in \B$ (meaning $\beta_j = 0$ for some $j$), so that $c$ is a coordinate of $w$.

The $i$-th generator of $\mathfrak m$ can be computed as follows:
If $[x^{\alpha_i}] \in \B$, also $c \cdot \lf_{\mathfrak m}(x^{\alpha_i})$ is a coordinate of $w$. 
Consequently, 
\begin{equation*}
\lf_{\mathfrak m}(x^{\alpha_i}) = \frac{c \cdot \lf_{\mathfrak m}(x^{\alpha_i})}{c}
\end{equation*}
is a ratio of coordinates of $w$.
This way, we get the generators $x^{\alpha_i} - \lf_{\mathfrak m}(x^{\alpha_i})$ of $\mathfrak m$ for all $i$ satisfying $[x^{\alpha_i}] \in \B$.

It remains to study the case when $[x^{\alpha_i}] \not \in \B$.
In this case we have 
\begin{equation}
\label{EquationComputingXjCoordinate}
[x^{\alpha_i}] = \sum_{j = 1}^r \lambda_j [x^{\beta_j}],
\end{equation}
for some scalars $\lambda_j \in \R$ (the linear combination can be extracted from the multiplication operator $\mathcal{M}_i$).
We know that $\frac{w}{c} = (\lf_{\mathfrak m}(x^{\beta_1}),\dots,\lf_{\mathfrak m}(x^{\beta_r}))$, hence we can use $\lf_{\mathfrak m}(x^{\beta_j})$ and plug it in the equation (\ref{EquationComputingXjCoordinate}) to get 
\begin{equation*}
\lf_{\mathfrak m}(x^{\alpha_i}) = \sum_{j = 1}^{r} \lambda_j \lf_{\mathfrak m}(x^{\beta_j}).
\end{equation*}
Thus, we can compute every generator of $\mathfrak m$.


